\documentclass[11pt]{amsart}
\usepackage{geometry}                
\usepackage{graphicx}
\usepackage{amssymb, amsthm, mathrsfs}
\usepackage{epstopdf}
\usepackage{enumerate}
\usepackage{color}
\usepackage{hyperref}
\usepackage{upgreek}
\usepackage{graphicx}
\usepackage{graphics, cancel}
\usepackage{color}
\usepackage{amsmath, amsthm, slashed}
\usepackage{amssymb}
\usepackage{rotating}
\usepackage{epsfig}
\usepackage{verbatim}
\usepackage{rotating}
\usepackage{graphicx}
\newcommand{\be}{\begin{equation*}}
\newcommand{\ee}{\end{equation*}}
\newcommand{\ben}[1]{\begin{equation}\label{#1}}
\newcommand{\een}{\end{equation}}
\newcommand{\bea}{\begin{eqnarray}}
\newcommand{\eea}{\end{eqnarray}}
\newcommand{\bean}{\begin{eqnarray*}}
\newcommand{\eean}{\end{eqnarray*}}

\renewcommand{\O}[1]{\mathcal{O}\left( #1 \right)}

\renewcommand{\H}{\underline{H}}
\newcommand{\Hh}{\mathscr{H}}
\newcommand{\Ll}{\mathscr{L}}

\renewcommand{\L}{\underline{L}}

\newcommand{\df}{\tilde{\nabla}}
\newcommand{\dfd}{\tilde{\nabla}^{\dagger}}

\newcommand{\scri}{\mathcal{I}}

\newcommand{\abs}[1]{\left|#1 \right|} 
\newcommand{\norm}[2]{\left|\left |#1 \right| \right |_{#2}} 
\newcommand{\ip}[3]{\left(#1,#2 \right )_{#3}} 
 
\newcommand{\Ub}{\overline{U}}

\newcommand{\tn}{\tilde{\nabla}}

\newcommand{\eq}[1]{(\ref{#1})}
\newtheorem{Theorem}{Theorem}

\newtheorem*{conj*}{Conjecture}
\newtheorem{Lemma}{Lemma}
\newtheorem{Corollary}[Theorem]{Corollary}
\newtheorem{Definition}{Definition}
\newtheorem{Proposition}{Proposition}
\numberwithin{Theorem}{section}
\numberwithin{Lemma}{section}
\title[Wave equation on AdS Black Holes]{Boundedness and Growth for the massive wave equation on asymptotically anti-de Sitter black holes}

\author{Gustav H. Holzegel}
\thanks{\texttt{g.holzegel@imperial.ac.uk} \\
\phantom{1   }\hspace{.05cm} Department of Mathematics, South Kensington Campus, Imperial College London, SW7 2AZ, UK \& \\
\phantom{1   }\hspace{.05cm} Department of Mathematics, Fine Hall, Washington Road, Princeton NJ 08544, USA}

\author{Claude M. Warnick}
\thanks{\vspace{.3cm} \texttt{warnick@ualberta.ca}\\
\phantom{1   }\hspace{.05cm} Department of Physics, 4-181 CCIS, University of Alberta, Edmonton AB T6G 2E, Canada}

\begin{document}
\begin{abstract}
We study the global dynamics of free massive scalar fields on general, globally stationary, asymptotically AdS black hole backgrounds with Dirichlet-, Neumann- or Robin- boundary conditions imposed on $\psi$ at infinity. This class includes the regular Kerr-AdS black holes satisfying the Hawking Reall bound $r_+^2 > |a|l$. We establish a suitable criterion for linear stability (in the sense of uniform boundedness) of $\psi$ and demonstrate how the issue of stability can depend on the boundary condition prescribed. In particular, in the slowly rotating Kerr-AdS case, we obtain the existence of linear scalar hair (i.e.\ non-trivial stationary solutions) for suitably chosen Robin boundary conditions. \\\phantom{1}\\ \phantom{1} \hfill ALBERTA THY 13-12
\end{abstract}
\maketitle
\vspace{.7cm}

\tableofcontents
\newpage

\section{Introduction}
The classical stability properties of asymptotically Anti de Sitter (aAdS) spacetimes have attracted recent attention in the general relativity community. To a large extent, this interest derives from the range of potential instability mechanisms which can be inferred -- so far only at the heuristic and numerical level  \cite{Anderson, DafHol, Bizon, Dias} -- from the geometry of these spacetimes, in particular from their asymptotic structure. These phenomena are entirely absent in the asymptotically flat case and have culminated in the conjecture that all asymptotically AdS spacetimes (including Kerr-AdS and pure AdS) may be unstable \cite{HolSmul}. See \cite{Dias:2012fu} for some recent work where this conjecture is being investigated.

From a classical perspective, the crucial feature of aAdS spacetimes is their failure of global hyperbolicity: Despite the fact that null-infinity is ``infinitely far away" in the sense that the affine length of null geodesics approaching infinity is indeed infinite, the causal structure of the spacetime also has the following property: Given a spacelike slice $\Sigma$, there exist points, $p$,  in $I^+\left(\Sigma\right)$ and (complete) past directed causal curves from $p$, which do not intersect $\Sigma$. This suggests that hyperbolic equations on such manifolds will, in general, require boundary conditions imposed at infinity to be well-posed.\footnote{See \cite{Friedrich} for an existence theorem for the full Einstein vacuum equation in this context. In particular, the above instability conjectures have to be supplemented by boundary conditions. In any case, the instability is believed to be present for all boundary conditions which ensure constant (finite) ADM mass at infinity.}

While a mathematical understanding of the potential non-linear instability mechanisms on aAdS spacetimes seems still out of reach, many results have been obtained for the linear massive wave equation
\begin{align} \label{vas}
\Box_g \psi + \frac{\alpha}{l^2} \psi = 0 \, ,
\end{align}
for $g$ an aAdS spacetime and $\alpha < \frac{9}{4}$ the Breitenlohner-Freedman bound\footnote{Our signature convention will be $(-+++)$ throughout.} imposed on the mass \cite{BF}.

In \cite{Holbnd}, the first author began to study the local and global properties of a particular class of solutions of (\ref{vas}), namely those satisfying ``Dirichlet"-conditions (or rather, the strongest possible radial decay for $\psi$). A well-posedness theorem was established for this class in \cite{Holwp, vasy}. Secondly, the boundedness of solutions on Schwarzschild-AdS and sufficiently slowly rotating Kerr-AdS backgrounds was proven \cite{Holbnd}. Later, in collaboration with J.~Smulevici, it was shown that the global solutions under consideration in fact decay logarithmically in time \cite{HolSmul}. The logarithmic decay rate is believed to be sharp and intimately connected to the geometric (trapping) properties of asymptotically AdS spacetimes. In fact a logarithmic rate has recently shown to be sharp \cite{HolSmul2}.

In \cite{Warnick:2012fi}, the second author established a far more general well-posedness theorem, which allowed for a wider range of boundary conditions (which can be imposed for $\frac{5}{4}<\alpha<\frac{9}{4}$) and in addition required less regularity of the solutions. The difficulty with handling the new boundary conditions (corresponding to less radial decay for $\psi$) arises from the fact that the usual $\partial_t$-energy fluxes for $\psi$ are infinite. This issue was successfully resolved by a renormalization scheme in \cite{Warnick:2012fi}, which adopts ideas of Breitenlohner and Freedman \cite{BF} but in fact works for any asymptotically AdS spacetime. The insight of \cite{Warnick:2012fi} is that if the equations and energies are expressed in terms of so-called ``twisted" or renormalized derivatives
\begin{equation} \label{twist}
\tilde{\nabla}_\mu \psi := f  \nabla_\mu \left( f^{-1} \psi\right),
\end{equation}
for an appropriate  ``twisting" function $f$, then the divergences at infinity disappear. Moreover, the energy density is positive for appropriate $f$, at least near infinity, which suffices for a well-posedness statement near the AdS boundary. For completeness, we state here a version of these results. We first define

\begin{Definition}
Let $(\mathcal{M}, g)$ be a time oriented Lorentzian manifold with an asymptotically AdS end, with asymptotic radial coordinate $r$ which we assume extends as a smooth positive function throughout $\mathcal{M}$ (see \S\ref{sec:genstat} for a sufficiently general definition of an aAdS end). Let $\Sigma$ be a spacelike surface which extends to the conformal infinity of the asymptotically AdS end, $\scri$. Let $n_\Sigma$ be the future directed unit normal of $\Sigma$ and define
\be
\hat{n}_\Sigma = r n_\Sigma,
\ee  
to be the rescaled normal. Let $f$ be a smooth positive function on $\mathcal{M}$ such that $f  r^{\frac{3}{2} - \kappa}= 1+\O{r^{-2}}$ as $r \to \infty$ for some $\kappa>0$, which will be related to $\alpha$ by $\alpha = 9/4-\kappa^2$. We denote by $\mathcal{D}^+(\Sigma)$ the region $D^+(\Sigma \cup (I^+(\Sigma) \cap \scri))$ which is the future Cauchy development of $\Sigma$ together with the portion of $\scri$ lying to the future of $\Sigma$.

We define the norms
\bean
\norm{\phi}{\L^2(\Sigma)}^2 &=& \int_\Sigma  \frac{\phi^2}{r} dS_{\Sigma},  \\
\norm{\phi}{\H^1(\Sigma, \kappa)}^2 &=& \int_\Sigma \left( |\tilde{\nabla} \phi |^2 + \frac{\phi^2}{r^2}\right) r dS_{\Sigma} \, ,
\eean
where the twisted derivative is defined as in (\ref{twist}) and we use the induced metric on $\Sigma$ to define $|\tilde{\nabla} \phi |^2$ and $dS_\Sigma$. We denote by $\H^1_0(\Sigma, \kappa)$ the completion in the $\H^1(\Sigma, \kappa)$ norm of the space of smooth functions supported away from $\scri$.

We furthermore say that a $C^1$ function $\phi$  on $\mathcal{M}$ obeys Dirichlet, Neumann or Robin boundary conditions if the following hold
\begin{enumerate}[i)]
\item Dirichlet:
\be
r^{\frac{3}{2} - \kappa} \phi \to 0, \quad \textrm { as } r \to \infty.
\ee
\item Neumann:
\be
r^{\frac{5}{2} + \kappa}\, \tilde{\nabla}_r \phi \to 0, \quad \textrm { as } r \to \infty.
\ee
\item Robin\footnote{While clearly the Robin condition includes the Neumann condition as a sub-case, it is convenient to follow the classical path of distinguishing the two.}:
\be
r^{\frac{5}{2} + \kappa}\, \tilde{\nabla}_r \phi + \beta r^{\frac{3}{2} - \kappa} \phi \to 0, \quad \textrm { as } r \to \infty,
\ee
where $\beta \in C^\infty(\scri)$.
\end{enumerate}
\end{Definition}

We can now state the well-posedness result. For the full details, including the precise weak formulations and the higher regularity and asymptotic conditions on the data, see \cite{Warnick:2012fi}.

\begin{Theorem}[Well Posedness]\label{WPThm}
\begin{enumerate}[1)]
\item Let $\uppsi\in \H^1_0(\Sigma, \kappa)$, $\uppsi' \in \L^2(\Sigma)$. Then there exists a unique $\psi$ such that $\psi|_{\Sigma} = \uppsi$, $\hat n_\Sigma \psi|_\Sigma = \uppsi'$ which, in a weak sense, solves
\be
\Box_g \psi + \frac{1}{l^2}\left(\frac{9}{4} - \kappa^2 \right) \psi = 0 \, ,
\ee
in $\mathcal{D}^+(\Sigma)$ with Dirichlet boundary conditions on $\scri$. If $\mathcal{S}$ is any spacelike surface in $\mathcal{D}^+(\Sigma)$ then $\psi|_\mathcal{S} \in \H^1_0(\mathcal{S}, \kappa)$, $\hat n_\mathcal{S}\psi|_\mathcal{S} \in \L^2(\mathcal{S})$.
\item Let $\uppsi\in \H^1(\Sigma, \kappa)$, $\uppsi' \in \L^2(\Sigma)$ and assume $0<\kappa<1$. Then there exists a unique $\psi$ such that $\psi|_{\Sigma} = \uppsi$, $\hat n_\Sigma \psi|_\Sigma = \uppsi'$ which, in a weak sense, solves
\be
\Box_g \psi + \frac{1}{l^2}\left(\frac{9}{4} - \kappa^2 \right) \psi = 0 \, ,
\ee
in $\mathcal{D}^+(\Sigma)$ with Neumann or Robin boundary conditions (for given $\beta$) on $\scri$. If $\mathcal{S}$ is any spacelike surface in $\mathcal{D}^+(\Sigma)$ then $\psi|_\mathcal{S} \in \H^1(\mathcal{S}, \kappa)$, $\hat n_\mathcal{S}\psi|_\mathcal{S} \in \L^2(\mathcal{S})$.
\end{enumerate}
If the initial conditions satisfy stronger regularity and asymptotic conditions\footnote{these conditions are of the form of those appearing in Theorem \ref{schbdn}}, then in fact $\psi|_\mathcal{S} \in H^k_{\textrm{loc.}}(\mathcal{S})$, $\hat n_\mathcal{S}\psi|_\mathcal{S} \in   H^{k-1}_{\textrm{loc.}}(\mathcal{S})$ for any integer  $k\geq 2$ and we obtain an asymptotic expansion
\be
\psi = \frac{1}{r^{\frac{3}{2}-\kappa}} \left[ \psi_0^-  +\O{r^{-1-\kappa}} \right ] + \frac{1}{r^{\frac{3}{2}+\kappa}} \left[ \psi_1^+ +  \O{r^{\kappa-1}}\right ].
\ee
The functions $\psi_i^\pm\in H^{k-1-i}(\scri)$ satisfy:
\be
\begin{array}{ccl}
\psi_0^-  =0&\qquad& \textrm{if $\psi$ satisfies Dirichlet boundary conditions,} \\ 
 \psi_1^+ =0&\qquad& \textrm{if $\psi$ satisfies Neumann boundary conditions,} \\
 2\kappa \psi_1^+ -\beta \psi_0^-=0  &\qquad& \textrm{if $\psi$ satisfies Robin boundary conditions.}
\end{array}
\ee
Thus for sufficiently regular initial data we obtain a classical solution to the initial boundary value problem.
\end{Theorem}

{ \bf Remark:} Note that while we say that ``$\psi$ satisfies Dirichlet conditions if $\psi r^{\frac{3}{2}-\kappa} \rightarrow 0$", we eventually establish stronger decay for $\psi$ ($\psi \sim r^{-\frac{3}{2} - \kappa}$) than what is implied by this condition. The point here is that the above Dirichlet condition suffices to eliminate the Neumann-branch of the solution. Conversely, the Neumann condition eliminates the Dirichlet branch. 
\newline 

Given a well posedness theorem of this generality, which in particular holds for the black hole spacetimes which we consider in this paper, we can enquire about the global behaviour of such solutions on black hole backgrounds and address the important question of how the stability properties depend on the boundary conditions imposed at infinity. This is the content of the present paper, which studies the dynamics of (\ref{vas}) on general asymptotically AdS backgrounds, which are globally stationary and contain a non-degenerate Killing horizon. Note that this class includes the regular Kerr-AdS black hole spacetimes (i.e.~$|a|<l$) which satisfy the Hawking-Reall bound ($r_{+}^2> |a|l$) on their parameters.

In turning from the local properties near the AdS boundary (exploited in the well posedness statement above) to global properties, one immediately faces the following difficulty: While it is relatively straightforward to see how the renormalization scheme removes the divergences at infinity, it is not at all clear whether this still yields a \emph{globally} non-negative energy on spacelike slices.\footnote{As mentioned above, positivity near infinity is immediate from the asymptotics of the twisting function.} In other words, it is not clear at all whether a global twisting function $f$ to achieve positivity exists\footnote{To retain the conservation property of the renormalized energy, that twisting function is required to be invariant under the stationary Killing field of the background.}, and if it exists, how it can be found. 

In this paper, we establish a simple criterion for (in)stability by relating the issue to the existence of a negative eigenvalue of a degenerate elliptic operator $L$. This operator emerges as follows. Pick a slice $\Sigma_0$ intersecting the future event horizon $\mathcal{H}^+$ and foliate the black hole exterior to the future of $\Sigma_0$ by slices $\Sigma_t$ arising as the push-forward of $\Sigma_0$ by the timelike isometry.  With $T$ denoting the stationary Killing vector (which is globally $\partial_t$ with respect to this slicing and timelike away from the horizon), we write the wave equation as
\begin{align} \label{elform}
L \psi  =  -TT\psi + B T\psi \, ,
\end{align}
where $B$ denotes a purely spatial operator that contributes only a boundary term in the following sense: If (\ref{elform}) is multiplied by $T\psi$ and integrated over a spacetime slab to produce the energy estimate, the term $BT\psi \cdot T\psi \sim B \left(\left(T\psi\right)^2 \right)$ will not contribute on $\Sigma_t$ but only on $\mathcal{H}^+$, where it has a good sign due to the non-vanishing positive surface gravity, and infinity, where it vanishes for Dirichlet or Neumann boundary conditions. The operator $L$ is an elliptic operator on the slices $\Sigma_t$ which degenerates in a precise way at the event horizon and at infinity. Rewriting $L$ in terms of twisted derivatives (where we twist both at the horizon and and infinity) it is possible to translate the analysis of $L$ to a regular elliptic Dirichlet/ Neumann/ Robin boundary value problem, with the only difference that the usual derivatives have been replaced by twisted derivatives. We finally establish that one may recover all of the standard elliptic spectral theory, formulated now with respect to the twisted derivatives and Sobolev spaces introduced in \cite{Warnick:2012fi}. In particular, $L$ can be shown to have a discrete spectrum bounded from below for all admissible boundary conditions. Remarkably, this result does not require a careful \emph{global} twisting but only the correct asymptotics near the horizon and infinity. Once the result is established, however, we will use the existence of a lowest eigenvalue and its associated eigenfunction to identify the latter as the ``optimal" twisting function.

Now, if $L$ has a negative eigenvalue, one can construct a growing solution for $\psi$. If the spectrum of $L$ is bounded away from zero, $\psi$ remains uniformly bounded, as twisting with the first eigenfunction of $L$ produces a non-negative energy.\footnote{We ignore here the degeneration of this energy at the horizon, which can be fixed using the redshift. These techniques (including commutation to estimate higher derivatives) have become standard \cite{Mihalisnotes} and will not be discussed in the paper. Note also that we expect bounded solutions to decay at least logarithmically following the arguments of \cite{HolSmul}.} Finally, if there is a zero mode for $L$ then general solutions cannot be expected to decay if the zero mode is present in the data. We can show that in this case solutions can grow at most like $t$. In summary, there is a neat and easy to compute -- at least numerically -- criterion for linear stability in the sense of boundedness of solutions. It may appear that this criterion depends on the choice of spacelike slicing, however one can show that the \emph{sign} of the lowest eigenvalue of $L$ is independent of the spacelike surface on which $L$ is constructed (see Corollary \ref{cor42}).

In practical applications, computing exactly the lowest eigenvalue may only be possible numerically. However, if one is not on the verge of an instability, it often suffices to twist with a function which is ``sufficiently similar" to the first eigenfunction in order to show the existence of a non-negative energy. We will see examples of this below.

For pedagogical reasons, we will begin the paper in Section \ref{sec:SchwSch} with a detailed treatment of the Schwarzschild-AdS case, for which almost everything can be carried out explicitly, including writing down a global twisting function which ensures stability. We have

\begin{Theorem}
Smooth\footnote{Of course, the actual statements are proven in some (twisted) Sobolev space, see Theorem \ref{schbdn}.} solutions $\psi$ to (\ref{vas}) on Schwarzschild-AdS remain uniformly bounded on the black hole exterior both in the case of Dirichlet- and Neumann boundary conditions imposed on $\psi$ at infinity.
\end{Theorem}
In particular, the theorem generalizes the result of \cite{Holbnd} to the boundary conditions of \cite{Warnick:2012fi}. Alternatively, one can consider Robin boundary conditions at infinity. For such boundary conditions we show that both stability and instability can occur depending on the choice of Robin function. In other words, Robin boundary conditions (while also ``reflecting") can make the black hole unstable. For a critical $\beta$, there must exist a zero mode, which is a manifestation of ``linear hair" in the language of the high energy physics community.\footnote{Here such solutions are typically found numerically by integrating out from the horizon. The critical $\beta$ is picked up automatically. See for example \cite{Hertog:2004ns}, where the non-linear problem is considered.}

After treating the Schwarzschild case in great detail, we proceed by introducing the general theory outlined above for arbitrary stationary black hole backgrounds in Section \ref{sec:genstat}. In the process we call on a Rellich-Kondrachov like compactness result whose proof we defer to Section 6. In  Section \ref{sec:KerrAdS} we illustrate the general theory by treating the case of Kerr-AdS black holes satisfying the Hawking-Reall bound and the regularity bound $|a|<l$. We obtain

\begin{Theorem} \label{theo:2}
The previous theorem applies also to Kerr-AdS backgrounds satisfying the Hawking Reall bound $r_{+}^2 > |a|l$ and for $\frac{|a|}{l}<\epsilon$ for some $\epsilon>0$. In the Dirichlet case we may take $\epsilon=1$.
\end{Theorem}
Again, we remark that there are Robin boundary conditions that lead to solutions which grow on the black hole exterior or, in the critical case, to non-trivial solitonic solutions. In the Dirichlet case, Theorem \ref{theo:2} strengthens the result of \cite{Holbnd}, which required $2|a|<l$, to the full range of admissible parameters.

 We note that our results justify the approach taken in \cite{Dias:2010ma}, where Kerr-AdS black holes admitting linear hair were sought numerically and it was argued that this was the threshold for an instability. Since in \cite{Dias:2010ma} the authors make use of Boyer-Linquist coordinates (which are not regular on the horizon) it is not immediately apparent that the appearance of linear scalar hair is a threshold for instability, however our analysis in terms of a regular slicing shows this to be the case.

\subsection*{Acknowledgements}
We would like to thank Mihalis Dafermos for helpful comments, and for reading a draft of this manuscript. We would also like to thank the organisers of the workshop ``On the problem of collapse in GR'' at the University of Miami, Jan.~2012, where some of this work was undertaken. GHH thanks Jacques Smulevici for discussions and acknowledges funding through NSF grant DMS-1161607. CMW acknowledges funding from PIMS and NSERC.

\newpage
\section{The AdS-Schwarzschild Black Hole} \label{sec:SchwSch}

Let $r_{+}$ be the largest root of the cubic
 \be
 r-2 m +\frac{r^3}{l^2}=0,
\ee
and define $(\mathcal{R}, g)$ to be the manifold with boundary
\be
\mathcal{R} = \mathbb{R}_{t\geq 0} \times \mathbb{R}_{r\geq r_{+}} \times S^2,
\ee
endowed with the metric
\ben{gdd}
ds^2 = -\left( 1-\frac{2 m}{r} +\frac{r^2}{l^2}\right)dt^2 + \frac{4m}{r}\frac{1}{1+\frac{r^2}{l^2}} dr dt + \frac{1+\frac{2 m}{r} +\frac{r^2}{l^2}}{\left(1+\frac{r^2}{l^2} \right)^2}dr^2 + r^2 d\Omega_2^2,
\een
where $d\Omega_2^2$ is the canonical metric on the unit $2$-sphere. We will also use the notation $\not{\hspace{-.1cm} g}=r^2 d\Omega_2^2$ to refer to the induced metric on the orbits of the $SO(3)$ symmetry group. We refer to $(\mathcal{R}, g)$ as the exterior of the AdS-Schwarzschild black-hole with mass $m$ and AdS radius $l$. The horizon is located at $r_{+}$. Unlike the usual Schwarzschild coordinates, these coordinates enjoy the property of being regular at the horizon, but are merely stationary rather than static. The components of the inverse metric are
\ben{guu}
\begin{array}{rclcrcl}
g^{tt} &=& - \frac{1+\frac{2 m}{r} +\frac{r^2}{l^2}}{\left(1+\frac{r^2}{l^2} \right)^2} , &\qquad&g^{tr} &=&  \frac{2m}{r}\frac{1}{1+\frac{r^2}{l^2}} ,\\
g^{rr} &=& 1-\frac{2 m}{r} +\frac{r^2}{l^2},&\qquad& g^{AB} &=& \not{\hspace{-.1cm} g}^{AB}.
\end{array}
\een
The volume element is
\be
d\eta = r^2 dt dr  d\omega,
\ee
where $d\omega$ is the volume element on the unit $2$-sphere. We need to define a few hypersurfaces in the manifold. 

\begin{figure}[htbp]
\begin{center}
\begin{picture}(0,0)%
\includegraphics{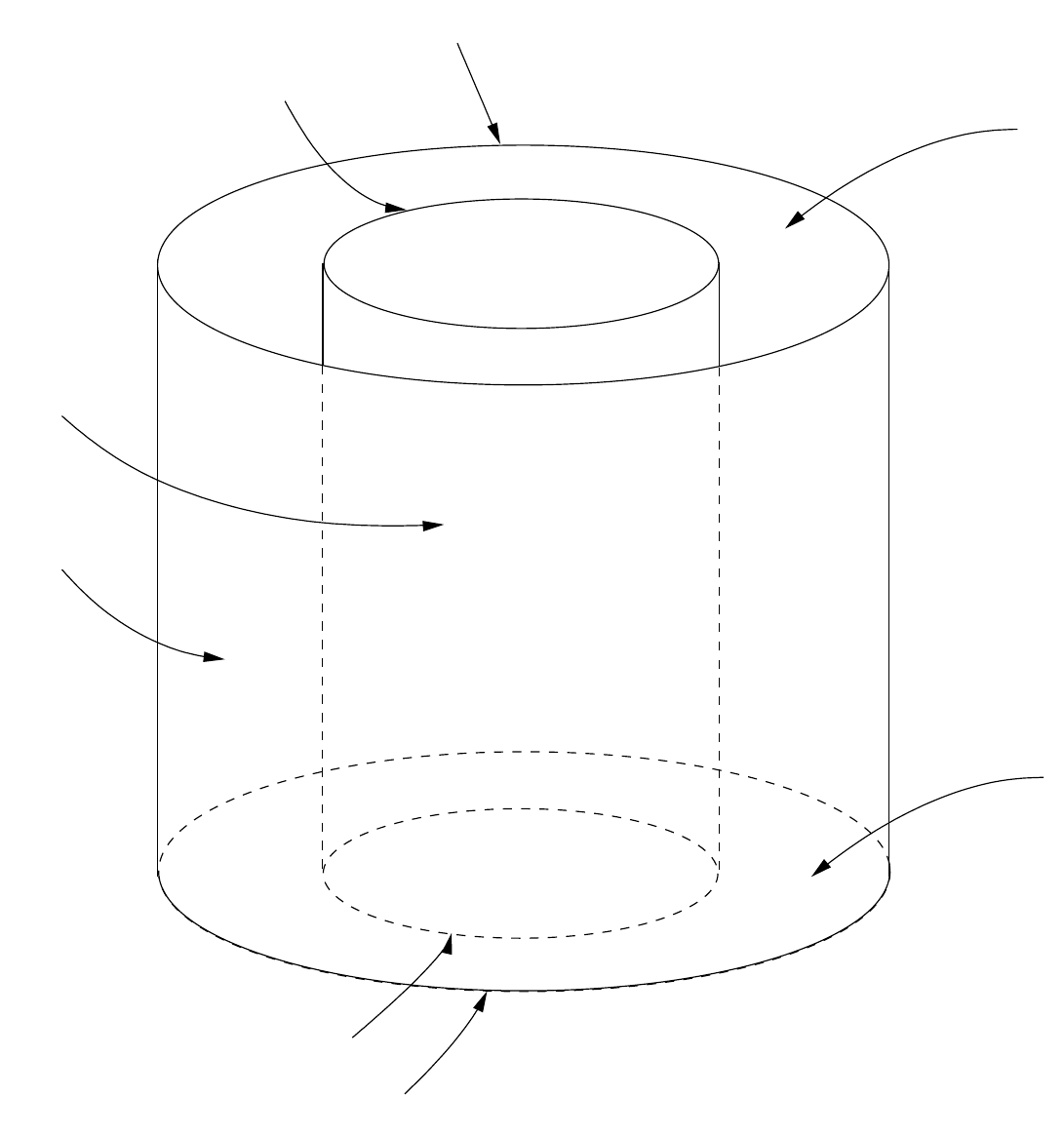}%
\end{picture}%
\setlength{\unitlength}{2960sp}%
\begingroup\makeatletter\ifx\SetFigFont\undefined%
\gdef\SetFigFont#1#2#3#4#5{%
  \reset@font\fontsize{#1}{#2pt}%
  \fontfamily{#3}\fontseries{#4}\fontshape{#5}%
  \selectfont}%
\fi\endgroup%
\begin{picture}(6974,7441)(1611,-7276)
\put(1757,-3419){\makebox(0,0)[lb]{\smash{{\SetFigFont{9}{10.8}{\rmdefault}{\mddefault}{\updefault}{\color[rgb]{0,0,0}$\tilde{\Sigma}_{R_2}^{[T_1, T_2]}$}%
}}}}
\put(8348,-680){\makebox(0,0)[lb]{\smash{{\SetFigFont{9}{10.8}{\rmdefault}{\mddefault}{\updefault}{\color[rgb]{0,0,0}$\Sigma_{T_2}^{[R_1, R_2]}$}%
}}}}
\put(8570,-4971){\makebox(0,0)[lb]{\smash{{\SetFigFont{9}{10.8}{\rmdefault}{\mddefault}{\updefault}{\color[rgb]{0,0,0}$\Sigma_{T_1}^{[R_1, R_2]}$}%
}}}}
\put(3689,-6803){\makebox(0,0)[lb]{\smash{{\SetFigFont{9}{10.8}{\rmdefault}{\mddefault}{\updefault}{\color[rgb]{0,0,0}$S^2_{T_1, R_1}$}%
}}}}
\put(3234,-404){\makebox(0,0)[lb]{\smash{{\SetFigFont{9}{10.8}{\rmdefault}{\mddefault}{\updefault}{\color[rgb]{0,0,0}$S^2_{T_2, R_1}$}%
}}}}
\put(3100,-1604){\makebox(0,0)[lb]{\smash{{\SetFigFont{9}{10.8}{\rmdefault}{\mddefault}{\updefault}{\color[rgb]{0,0,0}$B$}%
}}}}
\put(4399,-18){\makebox(0,0)[lb]{\smash{{\SetFigFont{9}{10.8}{\rmdefault}{\mddefault}{\updefault}{\color[rgb]{0,0,0}$S^2_{T_2, R_2}$}%
}}}}
\put(3993,-7198){\makebox(0,0)[lb]{\smash{{\SetFigFont{9}{10.8}{\rmdefault}{\mddefault}{\updefault}{\color[rgb]{0,0,0}$S^2_{T_1, R_2}$}%
}}}}
\put(1626,-2479){\makebox(0,0)[lb]{\smash{{\SetFigFont{9}{10.8}{\rmdefault}{\mddefault}{\updefault}{\color[rgb]{0,0,0}$\tilde{\Sigma}_{R_1}^{[T_1, T_2]}$}%
}}}}
\end{picture}%
\caption{The region of interest}
\label{Fig1}
\end{center}
\end{figure}
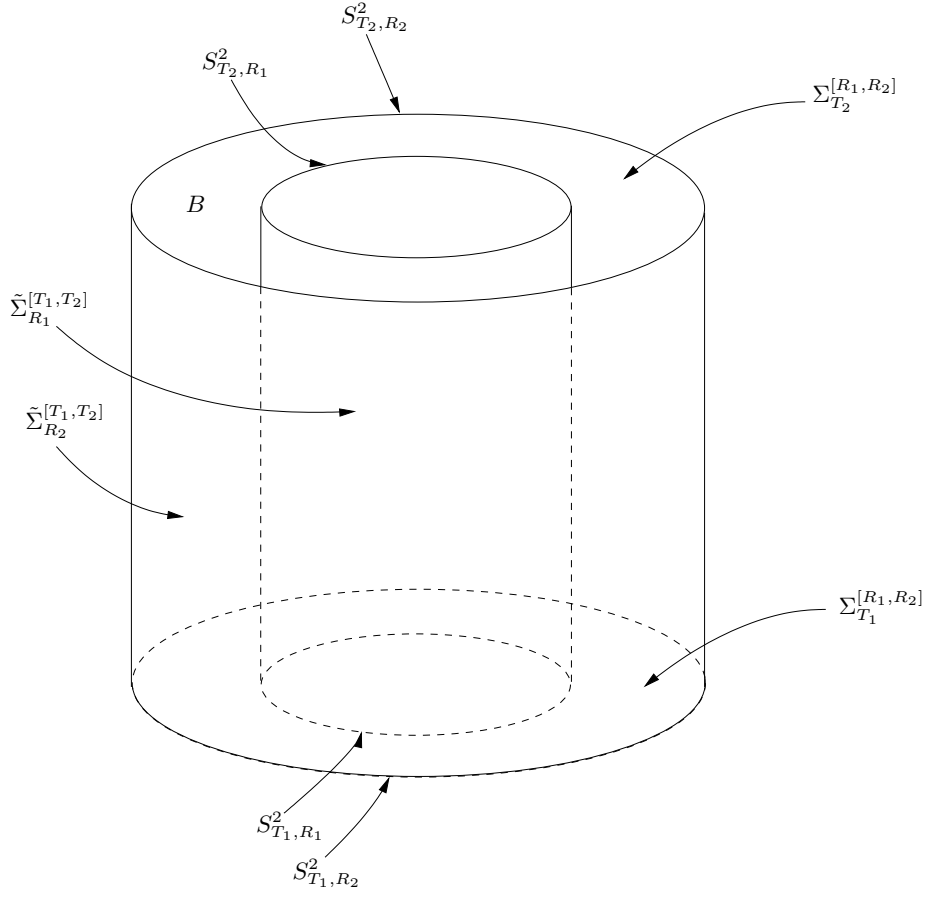

\begin{itemize}
\item $\Sigma_t$ denotes the hypersurface of constant $t$. It has a unit normal given by
\ben{ndef}
\begin{array}{rcl}
n &=& \sqrt{-g^{tt}} \frac{\partial}{\partial t} - \frac{g^{tr}}{\sqrt{-g^{tt}}}   \frac{\partial}{\partial r}, \\ 
n ^\flat &=& -\frac{1}{\sqrt{-g^{tt}}} dt,
\end{array}
\een
and an induced volume element
\ben{stvol}
dS_{\Sigma_t} = \sqrt{-g^{tt}} r^2 dr d\omega.
\een
Note that $\Sigma_t$ is a regular spacelike hypersurface, even as it approaches the horizon. we will denote the surface $\Sigma_t \cap\{r_1\leq r \leq r_2\}$ by $\Sigma_t^{[r_1, r_2]}$.
\item $\tilde{\Sigma}_r$ denotes the hypersurface of constant $r$. It has a unit normal
\ben{mdef}
\begin{array}{rcl}
m &=& \sqrt{g^{rr}} \frac{\partial}{\partial r} + \frac{g^{tr}}{\sqrt{g^{rr}}}   \frac{\partial}{\partial t}, \\ 
m^\flat &=& \frac{1}{\sqrt{g^{rr}}} dr,
\end{array}
\een
and an induced volume element
\ben{srvol}
dS_{\tilde{\Sigma}_r} = \sqrt{g^{rr}} r^2 dt  d\omega.
\een
Note that $m$ degenerates as we approach the horizon, this is because the normal to $\tilde{\Sigma}_r$ becomes null. The combination $m^\mu dS_{\tilde{\Sigma}_r}$ is well behaved however and gives the appropriate normal volume element on the horizon. We will denote the surface $\tilde{\Sigma}_r \cap\{t_1\leq t \leq t_2\}$ by $\tilde{\Sigma}_r^{[t_1, t_2]}$.
\item The surfaces $\Sigma_t$ and $\tilde{\Sigma}_r$ meet in the two-spheres $S^2_{t, r}$ which have induced volume element
\ben{s2vol}
dS_{S^2_{t,r}} = r^2 d \omega.
\een
\end{itemize}
Figure \ref{Fig1} shows the region which we shall consider, consisting of the solid annulus between two values of $r$ and two of $t$: \mbox{$B = \{(t, r, x^A) \in [T_1, T_2]\times [R_1, R_2]\times S^2\}$}.

We note that the metric \eq{gdd} has a Killing field which is given by
\ben{kill}
T = \frac{\partial}{\partial t}, \qquad  \qquad T^\flat =g^{tr} dr  -g^{rr}dt.
\een
\subsection{The untwisted energies}

Now we consider the wave equation \eq{vas}. This has an associated energy-momentum tensor given by
\ben{em}
T_{\mu \nu} = \nabla_\mu \psi \nabla_\nu \psi - \frac{1}{2}g_{\mu \nu}\left( \nabla_\sigma \psi \nabla^\sigma \psi - \frac{\alpha}{l^2} \psi^2 \right).
\een
We will sometimes write $T_{\mu\nu}[\psi]$ to emphasise the dependence on $\psi$. For arbitrary $\psi$, $T_{\mu\nu}$ satisfies
\ben{cons}
\nabla_\mu T^{\mu}{}_\nu = \left( \Box_g \psi + \frac{\alpha}{l^2} \psi \right) \nabla_\nu \psi,
\een
so that when $\psi$ solves \eq{vas}, the energy momentum tensor is conserved. For any vector field $V$, we define the currents
\ben{Jdef}
J^V_\mu[\psi] = V^\nu T_{\mu \nu}[\psi], \qquad K^V[\psi] = {}^V\pi_{\mu \nu} T^{\mu \nu}[\psi],
\een
where $ {}^V\pi_{\mu \nu}$ is the deformation tensor
\ben{deftens}
{}^V\pi_{\mu \nu} =  \nabla_{(\mu} V_{\nu)} =  \frac{1}{2}\left(\nabla_{\mu} V_{\nu} +\nabla_{\nu} V_{\mu} \right) =\frac{1}{2}(\mathcal{L}_V g)_{\mu \nu}.
\een
As a consequence of \eq{cons}, when $\phi$ satisfies \eq{vas} we have
\ben{cur}
\nabla^\mu J^V_\mu [\psi] = K^V[\psi].
\een
Applying the divergence theorem to the region $B$ above, we have the following identity:
\bea\nonumber
&&\int_{\Sigma_{T_2}^{[R_1,R_2]}} J^V_\mu n^\mu dS_{\Sigma_{T_2}}-\int_{\Sigma_{T_1}^{[R_1,R_2]}} J^V_\mu n^\mu dS_{\Sigma_{T_1}} \\&&+ \int_{\tilde{\Sigma}_{R_1}^{[T_1,T_2]}} J^V_\mu m^\mu dS_{\tilde{\Sigma}_{R_1}} -\int_{\tilde{\Sigma}_{R_2}^{[T_1,T_2]}} J^V_\mu m^\mu dS_{\tilde{\Sigma}_{R_2}} +\int_B K^V d\eta=0. \label{enid}
\eea
Note that the surfaces with timelike normal receive a sign flip relative to what one would expect from the divergence theorem in $\mathbb{R}^3$. We will choose $V=T$, so that the bulk term drops out as the deformation tensor of a Killing vector is identically zero.

\subsubsection{The fluxes}
We now examine the terms in the energy identity \eq{enid}. A straightforward, if tedious, calculation delivers the following result for the flux through the spacelike surfaces
\ben{egflux1}
\int_{\Sigma_{t}^{[R_1,R_2]}} J^T_\mu n^\mu dS_{\Sigma_{t}} = \frac{1}{2}\int_{\Sigma_t^{[R_1,R_2]}} \left(-g^{tt} \left(\nabla_t \psi \right)^2+g^{rr} \left(\nabla_r \psi \right)^2 +\left( \not\hspace{-.1cm} \nabla \psi\right)^2 - \frac{\alpha}{l^2} \psi^2 \right) r^2 dr d\omega,
\een
and for the flux through the timelike surfaces we have the following
\ben{egflux2}
\int_{\tilde{\Sigma}_{r}^{[T_1,T_2]}} J^T_\mu m^\mu dS_{\tilde{\Sigma}_{r}} = \int_{\tilde{\Sigma}_r^{[T_1,T_2]}}\left(g^{rt}  \left(\nabla_t \psi \right)^2+g^{rr} \left(\nabla_t \psi \right)\left(\nabla_r \psi \right)\right) r^2 dt d\omega.
\een
There are three main issues with the identity which we get on making use of these fluxes, namely
\begin{enumerate}[1)]
\item The $(\partial_r \phi)^2$ term degenerates at the horizon. This is a well known problem with the energies associated to the Killing field which defines a Killing horizon. The resolution is to make use of the redshift effect to estimate radial derivatives close to the horizon (see \cite{Mihalisnotes} and references therein).
\item For $\alpha>0$ the energy is not positive definite. This can be resolved with the aid of a Hardy inequality in the range $\alpha<\frac{9}{4}$ in which equation \eq{vas} is well posed with Dirichlet boundary conditions, cf.~\cite{Holbnd}.
\item For $\frac{5}{4}<\alpha<\frac{9}{4}$ the equation is also well posed with Neumann (or Robin) boundary conditions. Inserting the expected fall-off in this case into \eq{egflux1}, \eq{egflux2} we find that both fluxes will blow up as $R_2\to \infty$.
\end{enumerate}

\subsection{The twisted energies}

We will now show that by `twisting' the radial derivative it is possible to simultaneously resolve points $2), 3)$ when $\alpha$ is in the range $\frac{5}{4}<\alpha<\frac{9}{4}$. Let us recall the definition of the twisted derivative:
\be
\tn_\mu u = f \nabla_\mu\left( \frac{u}{f}\right).
\ee
We will assume that $f = f(r)$, so that the only derivatives affected by twisting are in the radial direction. We can re-write the radial derivative in \eq{egflux1} in terms of twisted derivatives. In so doing, we will introduce a term proportional to $\psi^2$ and one proportional to $\psi \partial_r \psi$. We integrate the second of these terms by parts to give surface terms, together with a modification of the $\psi^2$ term. The final result then is that
\bea\nonumber
\int_{\Sigma_{t}^{[R_1,R_2]}} J^T_\mu n^\mu dS_{\Sigma_{t}} &=& \frac{1}{2}\int_{\Sigma_t^{[R_1,R_2]}} \left(-g^{tt} \left(\nabla_t \psi \right)^2+g^{rr} \left(\tn_r \psi \right)^2 +\left( \not\hspace{-.1cm} \nabla \psi\right)^2 +V(r) \psi^2 \right) r^2 dr d\omega\\&&\quad\label{egflux3} +\int_{S^2_{t, {R_2}}} S(r) \psi^2 r^2 d\omega-\int_{S^2_{t, {R_1}}} S(r) \psi^2 r^2 d\omega,
\eea
where we have a modified potential term
\ben{Vdef}
V(r) =  -\frac{1}{r^2}\partial_r\left (r^2 g^{rr} \frac{f'(r)}{f(r)} \right) - g^{rr} \left(\frac{f'(r)}{f(r)}\right)^2- \frac{\alpha}{l^2},
\een
and the surface term is given by
\ben{Sdef}
S(r) = \frac{1}{2}  g^{rr} \frac{f'(r)}{f(r)}.
\een

Now let us look at the flux through the timelike surfaces \eq{egflux2}. Twisting the radial derivative here will introduce a term proportional to $\psi \partial_t \psi$, which we can integrate by parts in time. Since the metric is stationary, this simply introduces surface terms, so that the flux can be re-written
\bea \nonumber
\int_{\tilde{\Sigma}_{r}^{[T_1,T_2]}} J^T_\mu m^\mu dS_{\tilde{\Sigma}_{r}} &=& \int_{\tilde{\Sigma}_r^{[T_1,T_2]}}\left(g^{rt}  \left(\nabla_t \psi \right)^2+g^{rr} \left(\nabla_t \psi \right)\left(\tn_r \psi \right)\right) r^2 dt d\omega \\&& \quad \label{egflux4}+\int_{S^2_{T_2, r}} S(r) \psi^2 r^2 d\omega-\int_{S^2_{T_1, r}} S(r) \psi^2 r^2 d\omega.
\eea
The key point here is that this is \emph{the same} $S(r)$. In the energy identity \eq{enid}, the surface contributions coming from the spacelike and the timelike surfaces cancel. Let us define
\ben{twEdef}
\mathcal{E}(t; [R_1, R_2]) =  \frac{1}{2}\int_{\Sigma_t^{[R_1,R_2]}} \left(-g^{tt} \left(\nabla_t \psi \right)^2+g^{rr} \left(\tn_r \psi \right)^2 +\left( \not\hspace{-.1cm} \nabla \psi\right)^2 +V(r) \psi^2 \right) r^2 dr d\omega,
\een
and
\ben{twFdef}
\mathcal{F}(r; [T_1, T_2]) = \int_{\tilde{\Sigma}_r^{[T_1,T_2]}}\left(g^{rt}  \left(\nabla_t \psi \right)^2+g^{rr} \left(\nabla_t \psi \right)\left(\tn_r \psi \right)\right) r^2 dt d\omega .
\een
We have the energy identity
\ben{nenid}
\mathcal{E}(T_2; [R_1, R_2]) - \mathcal{E}(T_1; [R_1, R_2]) = \mathcal{F}(R_2; [T_1, T_2]) -\mathcal{F}(R_1; [T_1, T_2]),
\een
which follows from \eq{enid} after cancelling the contributions from the $S^2$ on the corners of the region. Now we wish to make a choice of $f$ so that the new energy identity \eq{nenid} resolves the problems $2), 3)$ of the unmodified energy identity. To resolve problem $2)$, we need to choose $f(r)$ so that $V(r)$ becomes positive. To resolve problem $3)$, we need to choose $f(r)$ so that the energy is finite for the Neumann fall-off, as we take $R_2 \to \infty$. Examining the asymptotics, we see that $3)$ can be resolved by choosing a function $f(r)$ which satisfies
\be
f(r) \sim r^{-\frac{3}{2} + \sqrt{\frac{9}{4}-\alpha}}(1+\O{r^{-2}}),
\ee
for large $r$. Some experimentation shows that an appropriate choice of $f$ to resolve $2)$ is to take
\ben{fdef}
f(r) = \frac{1}{r} \left(1+\frac{r^2}{l^2} \right)^{-\frac{1}{4}+\frac{1}{2} \kappa},
\een
where we have defined $\kappa = \sqrt{\frac{9}{4}-\alpha}$. With this choice of $f$, the function $V$ is given by
\be
V(r) = \frac{(1-2\kappa)^2}{4(l^2+r^2)} + m \frac{4l^4 + 8 l^2 r^2 + r^4 (3-2\kappa)^2}{2 r^3 (l^2+r^2)^2},
\ee
which is manifestly positive. It also decays like $r^{-2}$, a necessary condition for the energy to converge with Neumann fall-off. Note that this choice of $f$ is no good when $m =0$, i.e.\ for pure AdS, since it is singular at $r=0$, which is in the domain of outer communication when $m=0$. In pure AdS an appropriate choice of $f$ is $f(r) = \left(1+\frac{r^2}{l^2} \right)^{-\frac{3}{4}+\frac{1}{2} \kappa}$.

With $f$ as in \eq{fdef}, we can take the limit $R_2 \to \infty$, $R_1 \to r_{+}$ in the energy identity \eq{nenid}. The contribution from the inner boundary becomes
\be
\lim_{R_1\to r_{+}} \mathcal{F}(R_1; [T_1, T_2]) = \int_{\mathcal{H}_{[T_1, T_2]}} g^{rt}  \left(\nabla_t \psi \right)^2 dt d\omega =: F[T_1, T_2],
\ee
and we note that $g^{rt}>0$ on the horizon, so this term is non-negative. We find that the contribution from $\scri$ vanishes for either Dirichlet or Neumann fall-off at infinity:
\be
\lim_{R_2\to \infty} \mathcal{F}(R_2; [T_1, T_2]) = 0.
\ee
We now turn to Robin boundary conditions, where we assume $\partial_t \beta = 0$. We then find
\bean
\lim_{R_2\to \infty} \mathcal{F}(R_2; [T_1, T_2]) &=& -\lim_{r \to \infty} \int_{\tilde{\Sigma}_r^{[T_1,T_2]}}\frac{\beta}{2 l^2} \partial_t \left((r^{\frac{3}{2}-\kappa} \psi)^2 \right)dt d\omega \\ &=& \frac{1}{2 l^2}  \int_{S^2_{T_1, \infty}} (r^{\frac{3}{2}-\kappa} \psi)^2 \beta\ d\omega-\frac{1}{2 l^2}  \int_{S^2_{T_2, \infty}}  (r^{\frac{3}{2}-\kappa} \psi)^2 \beta\ d\omega,
\eean
where we abuse notation slightly and understand the term $r^{\frac{3}{2}-\kappa} \psi$ in the integrals over the spheres at infinity to mean $\lim_{r \to \infty} r^{\frac{3}{2}-\kappa} \psi$, which we know to belong to $L^2(S^2_{t,\infty})$ from the well posedness theorem above, Theorem \ref{WPThm}. Defining the energy to be
\bea \nonumber
E[t]&=&  \frac{1}{2}\int_{\Sigma_t} \left(-g^{tt} \left(\nabla_t \psi \right)^2+g^{rr} \left(\tn_r \psi \right)^2 +\left( \not\hspace{-.1cm} \nabla \psi\right)^2 +V(r) \psi^2 \right) r^2 dr d\omega\\ && \qquad +  \frac{1}{2 l^2}  \int_{S^2_{t, \infty}} (r^{\frac{3}{2}-\kappa} \psi)^2 \beta\ d\omega, \label{twEdef2}
\eea
where we include the surface term only for Robin boundary conditions, we have
\ben{scwenid}
E[T_2] = E[T_1] - F[T_1, T_2].
\een
Taking all of this together, we arrive at the following result:

\begin{Proposition}\label{enlem}
Suppose $\psi$ is a weak solution of \eq{vas} on the exterior of the AdS-Schwarzschild black hole, subject to either Dirichlet, Neumann or Robin boundary conditions with $\beta \geq 0$, $\partial_t\beta=0$. Then the energy $E(t)$ defined in \eq{twEdef2} is finite, non-increasing and positive definite.
\end{Proposition}

Combining this with a now standard argument involving the redshift \cite{Mihalisnotes} to deal with the degeneration at the horizon, together with an elliptic estimate and Sobolev embedding, we can deduce that\footnote{recall $\hat{n}_\Sigma = r n_\Sigma$}:

\begin{Theorem} \label{schbdn}
Let $\psi$ be a weak solution of
\be
\Box_g \psi + \frac{\alpha}{l^2} \psi = 0,
\ee
in $\mathcal{R}$, with initial conditions $\psi|_{\Sigma_0}=\uppsi$, $\hat n_{\Sigma_0} \psi|_{\Sigma_0}=\uppsi'$. Let $\kappa = \sqrt{\frac{9}{4}-\alpha}$.
\begin{enumerate}[1)]
\item \textbf{\emph{[Neumann and Robin conditions]}} Suppose $\frac{5}{4}<\alpha<\frac{9}{4}$,  $\psi$ satisfies either Neumann or Robin boundary conditions such that $E[t]$ is positive definite, and suppose furthermore that
\be
\sum_{i=0}^1\norm{\left(\hat n_{\Sigma_0} \right)^i\psi}{\H^1(\Sigma_0, \kappa)} + \norm{\left(\hat n_{\Sigma_0} \right)^2\psi}{\L^2(\Sigma_0)}  < \infty.
\ee
Then $\psi$ is locally $C^0$ with
\ben{schbdn1}
\sup_{\Sigma_t} \abs{r^{\frac{3}{2}-\kappa}\psi} \leq C \left(\sum_{i=0}^1\norm{\left(\hat n_{\Sigma_0} \right)^i\psi}{\H^1(\Sigma_0, \kappa)} + \norm{\left(\hat n_{\Sigma_0} \right)^2\psi}{\L^2(\Sigma_0)}  \right ),
\een
for some constant $C$ independent of $t$. 
\item \textbf{\emph{[Dirichlet conditions]}}  Suppose $\alpha<\frac{9}{4}$ and $\psi$ satisfies Dirichlet boundary conditions so that $E[t]$ is positive definite
\begin{enumerate}[i)]
\item Suppose that
\be
\sum_{i=0}^1\norm{\left(\hat n_{\Sigma_0} \right)^i\psi}{\H^1(\Sigma_0, \kappa)} + \norm{\left(\hat n_{\Sigma_0} \right)^2\psi}{\L^2(\Sigma_0)}  < \infty.
\ee
Then $\psi$ is locally $C^0$ and for any $\epsilon>0$, there exists a $C(\epsilon)>0$, independent of $t$, such that
\ben{schbdn2}
\sup_{\Sigma_t} \abs{r^{\frac{3}{2}-\epsilon}\psi} \leq C(\epsilon) \left(\sum_{i=0}^1\norm{\left(\hat n_{\Sigma_0} \right)^i\psi}{\H^1(\Sigma_0, \kappa)} + \norm{\left(\hat n_{\Sigma_0} \right)^2\psi}{\L^2(\Sigma_0)}  \right ),
\een
\item Suppose that
\be
\sum_{i=0}^2\norm{\left(\hat n_{\Sigma_0} \right)^i\psi}{\H^1(\Sigma_0, \kappa)} + \norm{\left(\hat n_{\Sigma_0} \right)^3\psi}{\L^2(\Sigma_0)}  < \infty.
\ee
Then $\psi$ is locally $C^1$. Furthermore there exists $C$, independent of $t$ such that
\ben{schbdn3}
\sup_{\Sigma_t} \abs{r^{\frac{3}{2}+\kappa}\psi} \leq C \left(\sum_{i=0}^2\norm{\left(\hat n_{\Sigma_0} \right)^i\psi}{\H^1(\Sigma_0, \kappa)} + \norm{\left(\hat n_{\Sigma_0} \right)^3\psi}{\L^2(\Sigma_0)}  \right ),
\een
for $\frac{5}{4} <\alpha<\frac{9}{4}$. If $\alpha\leq \frac{5}{4}$ then for each $\epsilon>0$, there exists $C(\epsilon)>0$, independent of $t$, such that
\ben{schbdn4}
\sup_{\Sigma_t} \abs{r^{\frac{5}{2}-\epsilon}\psi} \leq C(\epsilon) \left(\sum_{i=0}^2\norm{\left(\hat n_{\Sigma_0} \right)^i\psi}{\H^1(\Sigma_0, \kappa)} + \norm{\left(\hat n_{\Sigma_0} \right)^3\psi}{\L^2(\Sigma_0)}  \right ).
\een
\end{enumerate}
\end{enumerate}
\end{Theorem}
We note that we can of course use the equation to re-write $ \left(\hat n_{\Sigma_0}\right)^i \psi$ in terms of $\uppsi, \uppsi'$ for $i=2, 3$. In particular, note that the finiteness of the right hand side of \eq{schbdn1} implies $\uppsi \in H^2_{\textrm{loc.}}, \uppsi' \in H^1_{\textrm{loc.}}$. Finiteness of the right hand side of \eq{schbdn3} implies $\uppsi \in H^3_{\textrm{loc.}}, \uppsi' \in H^2_{\textrm{loc.}}$. This extra differentiability is required to get control over the stronger $r$-weight in (\ref{schbdn3}, \ref{schbdn4}). It is obvious from the conformally coupled case that this is necessary. By using higher order energies we can also gain uniform control over derivatives of $\psi$, and more decay for the case $\alpha<5/4$ with Dirichlet boundary conditions.
\subsection{The case of negative Robin function}

The condition on $\beta$ in Proposition \ref{enlem}  is clearly sufficient to ensure that $E(t)$ is positive. It is not, however, necessary. Making use of a twisted version of the trace theorem for Sobolev spaces (c.f. Lemma 4.2.1 in \cite{Warnick:2012fi}) one can improve this condition to $-\epsilon \leq \beta$ for some $\epsilon>0$. In this approach, however, it is difficult to get an explicit $\epsilon$, much less an optimal one. On the other hand, it is clear that if $\beta$ is permitted to be sufficiently negative then $E[t]$ will fail to be positive -- simply take $r^{\kappa-\frac{3}{2}}$ as a test function and $\beta$ to be a negative constant of large magnitude.

We shall now discuss a sharp criterion to determine whether the energy is positive definite for a given Robin function $\beta$. For clarity, we will restrict to the case that $\beta$ is constant. This has the advantage of being consistent with the full set of isometries of the AdS-Schwarzschild black-hole. We would like to know then under what circumstances
\begin{equation} \label{eq:he2}
B[u, u]:= \int_{\Sigma_0} \left(g^{rr} \left(\tn_r u \right)^2 +V(r) u^2 \right) r^2 dr d\omega+ \frac{\beta}{l^2}  \int_{S^2_{\infty}} (r^{\frac{3}{2}-\kappa} u)^2\ d\omega \geq 0
\end{equation}
holds, where we assume that $u \in \H^1(\Sigma_0, \kappa)$. We have dropped the angular term since we can always reduce the energy by averaging over angular directions. Clearly this condition is independent of the magnitude of $u$, so we are free to normalise $u$ such that
\ben{normu}
\int_{\Sigma_0} u^2  \abs{g^{tt}}r^2 dr d\omega = 1 \, .
\een
We can then re-state the question of positivity of the energy to the problem of determining the sign of $\lambda$, where
\begin{align} \label{lamdef}
\lambda := \inf_{\{u \in \H^1(\Sigma_0, \kappa): \textrm{ \eq{normu} holds} \}} B[u,u].
\end{align}
By the twisted trace theorem $B[u,u]$ is bounded below so this infimum exists and is finite. This looks very much like the Rayleigh-Ritz formulation for the least eigenvalue of an elliptic operator and in fact this is precisely what it is: we will show  below (in a more general setting, cf.~Proposition \ref{prop1}) that the ``twisted" eigenvalue problem associated with (\ref{eq:he2}), namely (see (\ref{twistd}) for a definition of the adjoint $\tilde{\nabla}_r^\dagger$)
\begin{equation} \label{eqtwist}
\tilde{\nabla}_r^\dagger \left(g^{rr} \tilde{\nabla}_r u \right) +V(r) u=  \omega  \frac{1+\frac{2 m}{r} +\frac{r^2}{l^2}}{\left(1+\frac{r^2}{l^2} \right)^2} u \, ,
\end{equation}
where we require that $u$ is regular at $r = r_{{hor}.}$ and that
\begin{equation} \label{eqbdycond}
u \sim u_0 \left( \frac{2 \kappa}{r^{\frac{3}{2}-\kappa} }+  \frac{\beta}{r^{\frac{3}{2}+\kappa} }\right)+ \O{\frac{1}{r^{\frac{5}{2}}}}\quad \textrm{ as } r \to \infty, \textrm{\ \ \ \  with \ \ \ $\kappa^2 = \frac{9}{4}-\alpha$} \, ,
\end{equation}
is a ``good" eigenvalue problem: In particular, there is a discrete spectrum and a lowest eigenvalue $\omega_1 \left(\beta\right)$.
With this being established, one may -- for concrete computation of the lowest eigenvalue -- return to the untwisted form of (\ref{eqtwist}),
\begin{equation} \label{eq:he1}
-\frac{1}{r^2} \frac{d}{dr} \left[ \left(r^2 - 2 m r + \frac{r^4}{l^2}\right) \frac{du}{dr} \right] - \frac{\alpha}{l^2} u = \omega  \frac{1+\frac{2 m}{r} +\frac{r^2}{l^2}}{\left(1+\frac{r^2}{l^2} \right)^2} u \, ,
\end{equation}
also equipped with boundary condition (\ref{eqbdycond}). Whichever form one prefers, from (\ref{twEdef2}) one clearly has

\begin{Proposition}
Let $\omega_1(\beta)$ be the least eigenvalue of the Sturm-Liouville problem (\ref{eqtwist}) (or equivalently (\ref{eq:he1})) with boundary condition (\ref{eqbdycond}).
If $\omega_1(\beta) >0$, then the energy $E[t]$ is positive definite for this choice of $\beta$ and Theorem \ref{schbdn} applies.
\end{Proposition}

\begin{figure}
\vspace{1cm}
\includegraphics[scale=.75]{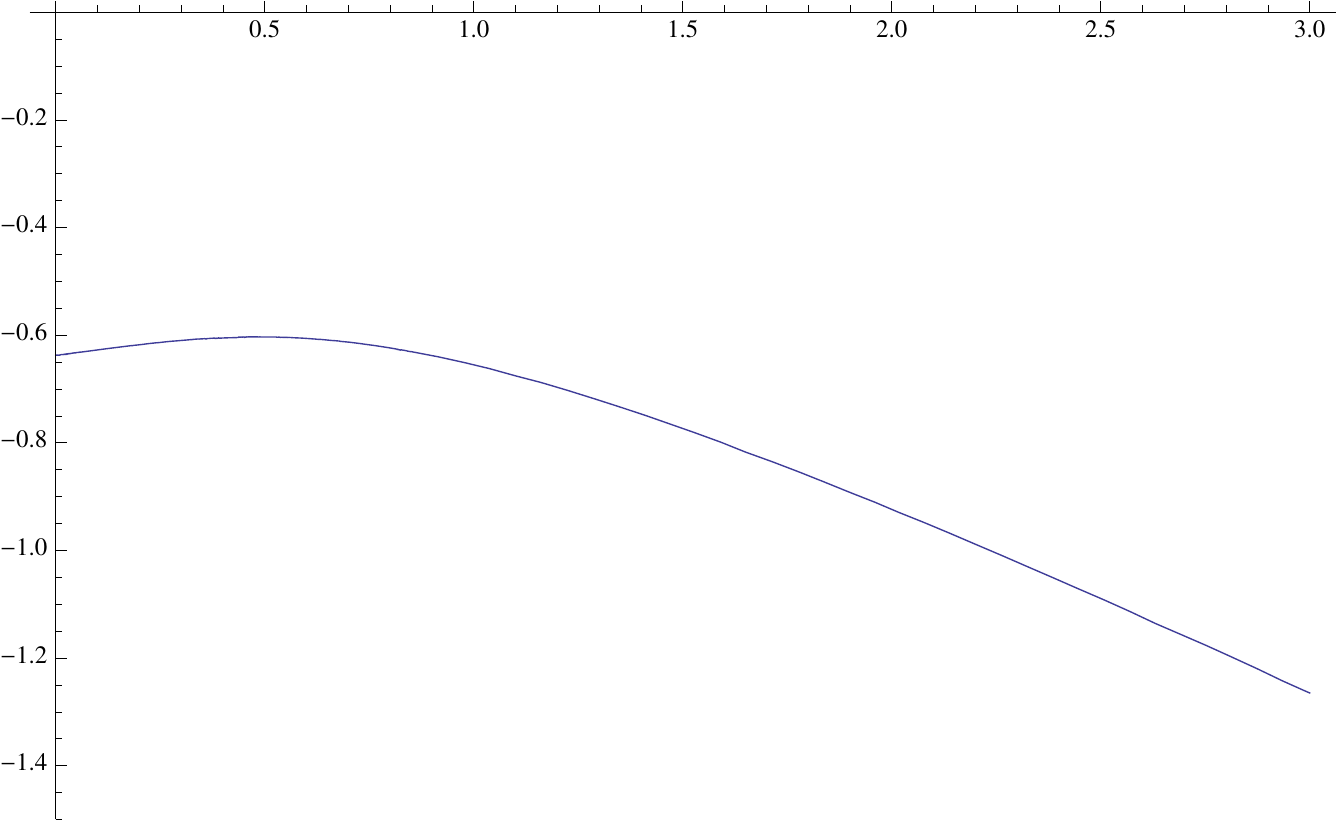}
\put(-150,195){$\frac{r_{+}}{l}$}
\put(-330,75){$\frac{\beta_c(r_{+})}{l}$}
\caption{The critical value of $\beta/l$ as a function of $r_{+}/l$ for the conformally coupled case, $\alpha = 2$ \label{betfig}.}
\end{figure}

If $\omega_1(\beta) <0$ then the energy may be negative and we shall see later that this implies an instability for the field, in the sense that there exist solutions which grow in time. We note that the choice of $\abs{g^{tt}}$ as a weight function in \eq{normu} is somewhat arbitrary, and related to the choice of spacelike slicing. We could choose any other positive function on $[r_{hor,}, \infty)$ with the same asymptotic behaviour and it would modify the Sturm-Liouville problem we obtain. The \emph{sign} of the lowest eigenvalue would be unchanged however, see Corollary \ref{cor42}.

In Figure \ref{betfig} we show a numerical plot of $\beta_c$, defined as the value of $\beta$ for which $\omega_1(\beta_c)=0$. We plot this as a function of the horizon radius $r_{+}$ for the conformally coupled equation with $\alpha = 2$. The value of $\beta_c$ at $r_{+}=0$ is consistent with the value $\beta_c/l = -2/\pi$ which can be determined for the exact anti-de Sitter spacetime. Provided $\beta>\beta_c$ we have boundedness of the wave equation with this boundary condition, whereas for $\beta<\beta_c$ we have an instability. We finally remark that for pure AdS this situation has 
been studied in \cite{WaldIshi}. 

\section{The twisted energy momentum tensor} \label{sec:twistEM}

In order to arrive at the final energy identity above,  \eq{scwenid}, we had to first consider the energy identity for a finite region, make some manipulations and then let the region become infinite. This worked because the divergences in the various energy fluxes of \eq{enid} `balanced one another out'. Very roughly speaking, we had the situation where
\be
E_1 \to \infty, \quad E_2 \to \infty, \qquad E_1-E_2 =0,
\ee
as we moved the boundary to infinity. We essentially found a $\delta E \to \infty$ such that $E_1 -\delta E$ and $E_2 - \delta E$ are bounded in the limit. It is convenient to have a means of constructing the renormalised fluxes directly, without having to first consider a finite region and then take a limit. In \cite{BF} this was achieved by a counter-term method for the case of pure anti-de Sitter space. While this method works for an asymptotically AdS space containing a horizon, as we discuss in Appendix \ref{counter}, it does have some drawbacks. Instead we shall introduce a  `renormalised' energy-momentum tensor by twisting, whose fluxes directly give the twisted energy identity we derived for AdS-Schwarzschild.

We start from the following observation:
\be
\frac{1}{f} \nabla_\mu \left[f^2 \nabla^\mu \left( \frac{\psi}{f}\right) \right] = \nabla_\mu \nabla^\mu \psi -\psi \frac{\nabla_\mu \nabla^\mu f}{f} ,
\ee
where $f$ is some smooth non-vanishing function. This can be checked by expanding the left hand side using the Leibniz rule. Motivated by this, we introduce two operators
\ben{twistd}
\tilde{\nabla}_\mu(\cdot) = f \nabla_\mu \left( \frac{1}{f} \cdot \right), \qquad \tilde{\nabla}^\dagger_\mu(\cdot) = -\frac{1}{f} \nabla_\mu \left( f \cdot \right).
\een 
We see that $\tn^\dagger_\mu$ is the formal adjoint of $\tn_\mu$ with respect to the natural $L^2$ inner product on the manifold. Whilst these operators are not derivations, they do commute with raising and lowering indices:
\be
\tn_\mu A_\sigma{}^{\tau_2 \ldots \tau_k} = g_{\sigma \tau_1} \tn_\mu A^{\tau_1 \tau_2 \ldots \tau_k},
\ee
and similarly for $\tn^\dagger$. We may re-write \eq{vas} in terms of the twisted covariant derivatives as
\bea \nonumber
0=\Box_g \psi + \frac{\alpha}{l^2} \psi &=& -\tn_\mu^\dagger \tn^\mu \psi + \left( \frac{\alpha}{l^2} +\frac{\nabla_\mu \nabla^\mu f}{f} \right) \psi \\ &=& -\tn_\mu^\dagger \tn^\mu \psi - V \psi, \label{twlap}
\eea
where we define $V$ in the second line. In the case above, we find that when $f=f(r)$, $V$ coincides with $V(r)$ as defined in \eq{Vdef}. For now, let us focus on \eq{twlap} without prejudice as to the underlying manifold or choice of $f$.

Motivated by the analogy with an untwisted problem, let us define the twisted energy-momentum tensor
\begin{Definition}
Given a smooth non-vanishing function $f$, with associated twisted derivative $\tilde{\nabla}_\mu$, we define the twisted energy-momentum tensor of the Klein-Gordon equation,
\be
\Box_g \psi + \frac{\alpha}{l^2} \psi=0,
\ee
to be the symmetric $2-$tensor
\ben{twem}
\tilde{T}_{\mu\nu}[\psi] = \tn_\mu \psi \tn_\nu \psi - \frac{1}{2} g_{\mu \nu} \left(\tn_\sigma \psi \tn^\sigma \psi + V \psi^2 \right),
\een
where
\be
V =  -\left( \frac{\nabla_\mu \nabla^\mu f}{f} +\frac{\alpha}{l^2} \right) .
\ee
\end{Definition}
This is not an energy momentum tensor in the usual sense as $\nabla_\mu \tilde{T}^{\mu \nu}[\psi] \neq 0$ in general for $\psi$ a solution of \eq{vas}. It does enjoy the following properties:

\begin{Proposition}[Properties of $\tilde{T}_{\mu\nu}$]
\begin{enumerate}[i)]
\item For a general $\phi\in C^2(\mathcal{M})$, we have
\ben{divt}
\nabla_\mu \tilde{T}^\mu{}_\nu[\phi] = \left( - \tn_\mu^\dagger \tn^\mu \phi - V \phi  \right) \tn_\nu \phi+ \tilde{S}_\nu[\phi],
\een
where
\be
\tilde{S}_\nu[\phi] = \frac{ \tn^\dagger_\nu (f V)}{2 f}  \phi^2 + \frac{ \tn^\dagger_\nu f}{2 f} \tn_\sigma \phi \tn^\sigma \phi.
\ee
\item Let $\psi$ be a $C^2$ solution of \eq{vas} and $X$ a smooth vector field. We define
\ben{twcur}
\tilde{J}^X_\mu[\psi] = \tilde{T}_{\mu \nu}[\psi] X^\nu, \qquad  \tilde{K}^X[\psi] ={}^X\pi_{\mu \nu} \tilde{T}^{\mu \nu}[\psi] + X^\nu \tilde{S}_\nu[\psi].
\een
Then
\ben{twcons}
\nabla^\mu \tilde{J}_\mu^X[\psi] = \tilde{K}^X[\psi].
\een
\item Suppose that $f$ is such that $V \geq 0$. Then $\tilde{T}_{\mu \nu}[\psi]$ satisfies the Dominant Energy Condition. In other words, if $X$ is a future pointing causal vector field, then so is $-\tilde{J}^X[\psi]$. 
\end{enumerate}
\end{Proposition}

Of key importance here is that $\tilde{S}_\nu$ and hence $\tilde{K}_\nu$ depend only on the $1$-jet of $\psi$, i.e. $\psi$ and $\tilde{\nabla}_\mu \psi$.  As a result, $\tilde{J}_\mu^X[\psi]$ is a \emph{compatible current} in the sense of Christodoulou \cite{orange}. Note that if $Y$ is a Killing vector of $g$ which preserves $f$, i.e. $\mathcal{L}_Y(f)=0$, then $\tilde{J}^Y_\mu[\psi]$ is a conserved current. 

In the AdS-Schwarzschild case considered above, $T$ preserves $f$. A very straightforward calculation then establishes that the flux through a spacelike surface is given by
\ben{twflux1}
\int_{\Sigma_t^{[R_1, R_2]}} \tilde{J}^T_\mu n^\mu dS_{\Sigma_t} =  \mathcal{E}(t; [R_1, R_2]),
\een
while for a timelike surface we find
\ben{twflux2}
\int_{\tilde{\Sigma}_{r^{[T_1, T_2]}}} \tilde{J}^T_\mu m^\mu dS_{\tilde{\Sigma}_r} = \mathcal{F}(r; [T_1,T_2]).
 \een
Where $\mathcal{E}, \mathcal{F}$ are as defined in (\ref{twEdef}, \ref{twFdef}) Thus the fluxes (\ref{twflux1}, \ref{twflux2}) together with the energy identity arising from integrating \eq{twcons} give precisely the twisted energy identity \eq{scwenid}. The advantage of introducing the renormalised energy momentum tensor is that all of the fluxes are finite as defined, so we may work with the energy identity on an infinite slab, without having first to consider a finite problem in order to regularise.
\section{Boundedness for general stationary black holes} \label{sec:genstat}

In this section we are going to establish a sharp criterion to determine whether solutions to the Klein-Gordon equation outside a given stationary, aAdS, black hole are bounded in time or not. We will assume that the spacetime has one asymptotically AdS end, one non-degenerate Killing horizon, and no other horizons or infinities. Our results extend easily to multiple horizons and multiple aAdS ends, and to higher dimensions, but we shall not pursue this possibility. We start by defining an asymptotically anti-de Sitter end in such a way that the well posedness result, Theorem \ref{WPThm}, holds as stated.

\begin{Definition}
Let $X$ be a manifold with boundary $\partial X$, and $g$ be a smooth Lorentzian metric on $int(X)$. We say that a connected component $\scri$ of $\partial X$ is an \emph{asymptotically anti-de Sitter end of $(int(X), g)$ with radius $l$} if:
\begin{enumerate}[i)]
\item There exists a smooth function $r$ such that $r^{-1}$ is a boundary defining function for $\scri$.
\item There exist coordinates, $x^\alpha$,  on the slices $r=\textrm{const.}$ such that we have locally
\be
g_{rr} = \frac{l^2}{r^2}+ \O{\frac{1}{r^4}},\qquad g_{r\alpha} = \O{\frac{1}{r^2}}, \qquad g_{\alpha\beta} = r^2 \mathfrak{g}_{\alpha \beta} + \O{1},
\ee
where $\mathfrak{g}_{\alpha\beta}dx^\alpha dx^\beta$ is a Lorentzian metric on $\scri$.
\item $r^{-2} g$ extends as a smooth metric on a neighbourhood of $\scri$.
 \end{enumerate}
We say that $r$ is the asymptotic radial coordinate and $\scri$ is the conformal infinity of this end.
\end{Definition}
Note that $r$ and $\mathfrak{g}$ are not unique. A different choice of $r$ gives rise to a different $\mathfrak{g}$ conformally related to the first. Condition $ii)$ can be weakened to $g_{rr} = l^2 r^{-2} + \O{r^{-3}}, g_{\alpha\beta} = r^2 \mathfrak{g}_{\alpha\beta} + \O{r}, g_{r\alpha} = \O{\frac{1}{r}}$ for well posedness of the massive wave equation\footnote{In fact, this is the condition imposed in \cite{vasy}. While such generalized spacetimes have less prominence in the physics literature, they exhibit interesting propagation of singularities studied in \cite{Pham}. In the Riemannian setting conformally compact manifolds with these asymptotics have been extensively studied, e.g.\ \cite{MR837196, Mazzeo1987260, MR1112625, MR2421542}.}, however one then needs to make a more careful choice of twisting function $f$, i.e. $f r^{3/2-\kappa} \sim 1 + f_1 r^{-1} + \O{r^{-2}}$, for a specific choice of $f_1$ determined by the metric functions. This isn't necessary for the purposes of the metrics we wish to consider here however. Condition $iii)$, sometimes known as weak asymptotic simplicity, is also not necessary for the well posedness of the massive wave equation\footnote{$C^2$ extensibility certainly suffices, $C^{1, \gamma}$ is probably enough}, but is necessary if one wishes to have a full asymptotic expansion for the scalar field near $\scri$. 

Motivated by the discussion of the AdS-Schwarzschild black hole in Section \ref{sec:SchwSch}, we now  introduce the notion of an asymptotically anti-de Sitter black hole. 

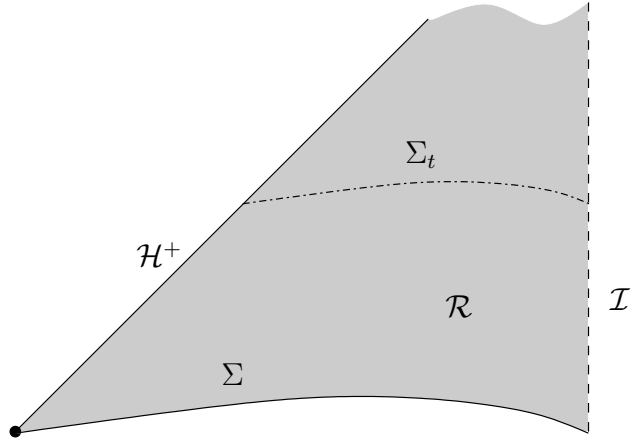
\begin{figure}[b]
\begin{picture}(0,0)%
\includegraphics{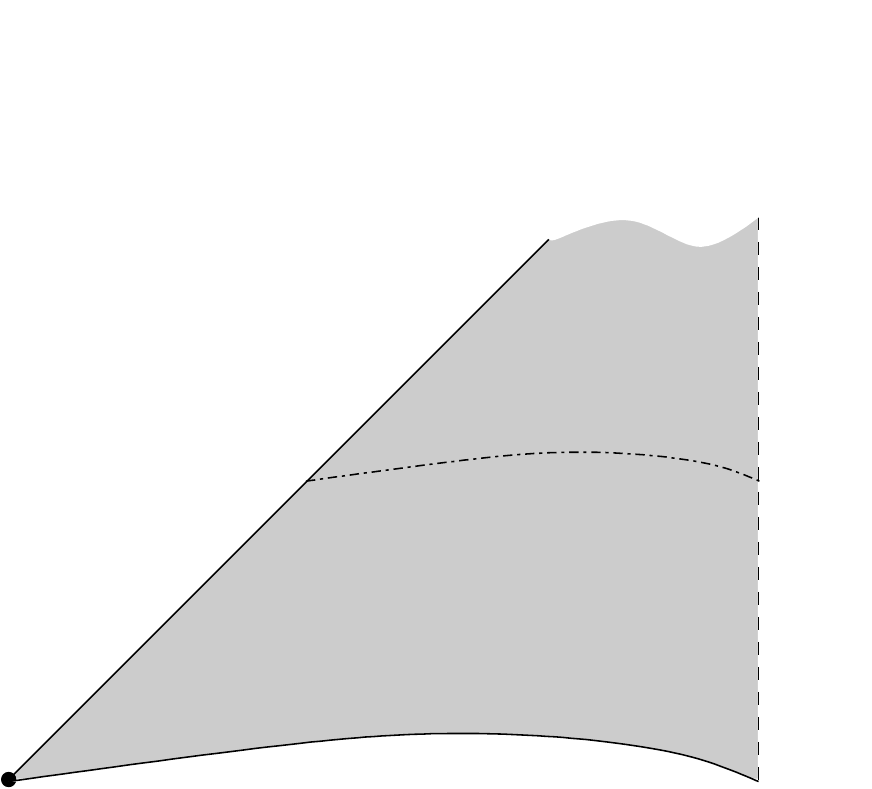}%
\end{picture}%
\setlength{\unitlength}{3947sp}%
\begingroup\makeatletter\ifx\SetFigFont\undefined%
\gdef\SetFigFont#1#2#3#4#5{%
  \reset@font\fontsize{#1}{#2pt}%
  \fontfamily{#3}\fontseries{#4}\fontshape{#5}%
  \selectfont}%
\fi\endgroup%
\begin{picture}(4183,3781)(2360,-5202)
\put(5103,-4440){\makebox(0,0)[lb]{\smash{{\SetFigFont{12}{14.4}{\rmdefault}{\mddefault}{\updefault}{\color[rgb]{0,0,0}$\mathcal{R}$}%
}}}}
\put(4850,-3473){\makebox(0,0)[lb]{\smash{{\SetFigFont{12}{14.4}{\rmdefault}{\mddefault}{\updefault}{\color[rgb]{0,0,0}$\Sigma_t$}%
}}}}
\put(3697,-4866){\makebox(0,0)[lb]{\smash{{\SetFigFont{12}{14.4}{\rmdefault}{\mddefault}{\updefault}{\color[rgb]{0,0,0}$\Sigma$}%
}}}}
\put(6137,-4393){\makebox(0,0)[lb]{\smash{{\SetFigFont{12}{14.4}{\rmdefault}{\mddefault}{\updefault}{\color[rgb]{0,0,0}$\scri$}%
}}}}
\put(3163,-4120){\makebox(0,0)[lb]{\smash{{\SetFigFont{12}{14.4}{\rmdefault}{\mddefault}{\updefault}{\color[rgb]{0,0,0}$\mathcal{H}^+$}%
}}}}
\end{picture}%
\caption{A schematic Penrose diagram for a stationary, asymptotically anti-de Sitter, black hole space time. Each point represents a compact $2$-surface diffeomorphic to $\mathcal{K}$.\label{Pendiag}}
\end{figure}

\begin{Definition}\label{adsbh}
We say that $(\mathcal{R}, \mathcal{H}^+, \Sigma, g, r, T)$ is a stationary, asymptotically anti-de Sitter, black hole space time with AdS radius $l$ if the following holds
\begin{enumerate}[i)]
\item $\mathcal{R}$ is a four dimensional manifold with stratified boundary $\mathcal{H}^+\cup \Sigma$, where $\mathcal{H}^+, \Sigma$ are themselves manifolds with compact, connected, common boundary $\mathcal{K}$.
\item  $\mathcal{R}$ is diffeomorphic to  $[0, \tau) \times [0, \rho) \times \mathcal{K}$, with $\mathcal{H^+} \simeq [0, \tau) \times \{0\} \times \mathcal{K}$ and
$\Sigma \simeq \{0\} \times [0, \rho) \times \mathcal{K}$.
\item $\Sigma$ is everywhere spacelike with respect to $g$, whereas $\mathcal{H}^+$ is null.
\item The spacetime $(\mathcal{R}, g)$ has an asymptotically anti-de Sitter end,  of AdS radius $l$, with conformal infinity $\scri \simeq [0, \tau) \times \{\rho\} \times \mathcal{K}$ and asymptotic radial coordinate $r$ and such that
\be
\mathcal{R} = D^+(\Sigma \cup \scri).
\ee 
We assume $r$ extends to a smooth positive function throughout $\mathcal{R}$.
\item $T$ is a Killing field of $g$ which is timelike on $\mathcal{R}\setminus\mathcal{H}^+$.
\item $T$ is normal to $\mathcal{H}^+$, and $r^{-1} T$ is uniformly bounded in length and tangent to $\scri$. Thus $\mathcal{H}^+$ is a Killing horizon generated by
 $T$, which we assume to be a non-extremal black hole horizon.
 \item If $\varphi^T_t$ is the one-parameter family of diffeomorphisms generated by $T$, then $\mathcal{R}$ is smoothly foliated by  $\varphi_t(\Sigma):=\Sigma_t$, $t\geq 0$.
\end{enumerate}
\end{Definition}
Note that our definition implies a compact horizon topology and so excludes those AdS black holes with  infinite  planar or hyperbolic horizons. We illustrate our definition with a schematic Penrose diagram in Figure (\ref{Pendiag}). We claim that the AdS-Schwarzschild and AdS-Kerr black hole (satisfying the Hawking-Reall bound) contain regions which satisfy this definition. See the introduction to \S\ref{sec:KerrAdS} for a more detailed discussion of this point. We need not restrict to these spacetimes, however. Our methods apply equally to spacetimes such as Kerr-Newman-AdS, the Schwazschild-AdS solutions with toroidal or compact hyperbolic symmetry orbits and even to more exotic spacetimes such as those exhibited in \cite{Martinez:2004nb}.

We note that the Klein-Gordon equation, \eq{vas}, with $\alpha<\frac{9}{4}$, is well posed on such a spacetime with initial data specified on $\Sigma$ and appropriate boundary conditions imposed at  $\scri$, cf.~Theorem \ref{WPThm}.

\subsection{Decomposing the metric}

We will make use of a $3+1$ decomposition of the metric which separates out the $t$ direction. This can be done in more than one way, but for us it will be convenient to make use of an ADM decomposition \cite{Arnowitt:1962hi} (sometimes referred to as Painlev\'e-Gullstrand type coordinates). We define various objects associated to the slicing $\Sigma_t$. Firstly, we denote
\be
\sigma:=-g(T,T) \, .
\ee
We know that $\sigma\geq0$,  vanishing only the horizon. We define $N$ to be the future directed unit normal to $\Sigma_t$
\be
N := n_{\Sigma_t}.
\ee
We define $h$ to be the induced metric on $\Sigma_t$. The lapse $A$ is defined to be
\be
A := g(N, T)^2,
\ee
and the shift vector $W$ is
\be
W := T - \sqrt{A} N.
\ee
By construction $W$ is orthogonal to $N$ and so is tangent to $\Sigma_t$.

Let $U$ be a coordinate patch of $\Sigma$ on which we define coordinates $\tilde{x}^i:U \to \mathbb{R}^3$. Setting $U_t = \varphi_t(U)$, we define $x^i:U_t \to \mathbb{R}^3$ by $x^i = \tilde{x}^i \circ (\varphi_t)^{-1}$. We then have that $(t, x_i)$ define coordinates on $\cup_{t\geq 0} U_t$, and taking an atlas of charts for $\Sigma$ we can cover $\mathcal{R}$ with such coordinate patches. In one of these local coordinate patches, we may write
\be
T = \frac{\partial}{\partial t}, \qquad W = W^i \frac{\partial}{\partial x^i},\qquad  h = h_{ij} dx^i dx^j ,
\ee
and the metric has the local form
\ben{statg}
g = -A dt^2 + h_{ij} (dx^i + W^i dt)(dx^j + W^j dt) \, ,
\een
where we know that $A$, $h_{ij}$ and $W^i$ are independent of $t$. The determinant is given by
\ben{detg}
\sqrt{\abs{\det{g}}} = \sqrt{A\det h}.
\een
We define $h^{ij}$ to be the matrix inverse of $h_{ij}$. The metric in this form can be conveniently inverted to give
\ben{statginv}
g^{-1} = -\frac{1}{A} \left(\frac{\partial}{\partial t}- W^i \frac{\partial}{\partial x^i} \right)\otimes  \left(\frac{\partial}{\partial t}- W^i \frac{\partial}{\partial x^i} \right)+ h^{ij} \frac{\partial}{\partial x^i} \otimes  \frac{\partial}{\partial x^j} .
\een

Now, the requirement that $\Sigma_t$ is spacelike implies that $A>0$ and $h_{ij}$ must be positive definite. These conclusions hold up to and on the horizon, since $N\propto dt^\sharp$ remains timelike there.  Next, the condition that $\partial_t$ be timelike implies that both \mbox{$\sigma=A-W^i W^j h_{ij} > 0$} and that $h^{ij}-\frac{1}{A}W^i W^j$ is positive definite. Now this condition holds everywhere outside the horizon, however we know that $\partial_t$ becomes null at the horizon. When this happens, $A-W^i W^j h_{ij}$ vanishes and $h^{ij}-\frac{1}{A}W^i W^j$ acquires a kernel corresponding to $h_{ij} W^j$. The remaining eigenvalues of $h^{ij}-\frac{1}{A}W^i W^j$ remain positive. Let us now consider how $A-W^i W^j h_{ij}$ vanishes. To do so, we make use of the surface gravity, $\varkappa$, of the Killing horizon. This is defined by the relation, evaluated on the horizon:
\ben{sgrav}
k^\mu \nabla_\mu k^\nu = \varkappa k^\nu.
\een
Where $k^a$ is the null generator of the Killing horizon. In our case, this is $\partial_t$. Since $\mathcal{H}^+$ is assumed to be a non-degenerate black hole horizon, we have that $\varkappa>0$. Making use of Killing's equation we can re-write \eq{sgrav} as
\begin{equation*}
-\partial_\mu (k_\nu k^\nu) = 2 \varkappa k_\mu.
\ee
The $t$ component of this equation is satisfied trivially, while the remaining components give 
\ben{sgrav2}
\partial_i (A- W^j W^k h_{jk}) = 2 \varkappa  W^j h_{ij}.
\een
Now, the right hand side is non-vanishing as a result of our definition: the surface gravity is positive for a non-extremal horizon. Furthermore $W^i$ cannot vanish at the horizon, since this would imply $T$ is parallel to $N$, however by construction $T$ is null and $N$ timelike. We conclude that $W^i\partial_i$ is normal to the horizon (with respect to the induced metric $h_{ij}$).

\subsection{Decomposing the wave operator}

Corresponding to the decomposition of the metric in the previous subsection, we have a $3+1$ decomposition of the Klein-Gordon equation. We can directly infer from \eq{detg} and \eq{statginv} that the wave operator locally takes the form
\ben{weq1}
\Box_g \psi = -\frac{1}{A} \frac{\partial^2 \psi}{\partial t^2}+ \frac{W^i}{A} \frac{\partial^2 \psi}{\partial t \partial x^i}+\frac{1}{\sqrt{A h}} \frac{\partial}{\partial x^i}\left( \sqrt{\frac{ h}{A}} W^i \frac{\partial \psi }{\partial t}\right)+ \frac{1}{\sqrt{A h}} \frac{\partial}{\partial x^i}\left( \sqrt{\frac{ h}{A}} a^{ij} \frac{\partial \psi }{\partial x^j}\right) \, ,
\een
where we introduce the tensor $a \in T\Sigma\odot T\Sigma$, given locally by $a^{ij} = A h^{ij} - W^i W^j$. Ignoring for a moment the terms involving time derivatives, the key point  here is that the purely spatial part of the operator \emph{is elliptic}, since $a^{ij}$ is positive definite away from the horizon. The ellipticity degenerates on the horizon, but in a controlled fashion which will allow us to handle the horizon as a boundary. We shall in fact be interested in the massive wave equation \eq{vas}. It will be convenient to write this in the following form
\ben{weq}
TT\psi - B T\psi+ L \psi = 0 ,
\een
where $B$ and $L$ are geometric differential operators defined on $\Sigma$. They are given in local coordinates by
\ben{Bdef}
B u = {W^i} \frac{\partial u}{\partial x^i}+\sqrt{\frac{A}{ h}} \frac{\partial}{\partial x^i}\left( \sqrt{\frac{ h}{A}} W^i u\right),
\een
and
\ben{Ldef}
L u = - \sqrt{\frac{A}{ h}} \frac{\partial}{\partial x^i}\left( \sqrt{\frac{ h}{A}} a^{ij} \frac{\partial u }{\partial x^j}\right) - \frac{\alpha}{l^2}A u \, .
\een
As an aside, note that we can immediately see that the equation in this form will have a nice conserved energy associated to it, since on multiplying by $\psi_t \sqrt{h/A}$ and integrating over $\Sigma_t$, the term coming from $B$ gives a pure boundary term which vanishes at infinity and gives the flux across the horizon. We get a good sign for this since $W^i$ is always outwards\footnote{i.e. directed towards the interior of $\mathcal{R}$} directed for a black hole horizon, a consequence of \eq{sgrav2}, together with the fact that $\varkappa >0$. This is no surprise since we have chosen our definition in such a way that the horizon is a \emph{black hole} horizon, rather than a \emph{white hole} horizon. As such, we expect the effect of the horizon to be to extract energy from disturbances propagating in the black hole exterior. The boundary term from integrating the term $\psi_t L \psi$ vanishes at the horizon and, provided we make suitable assumptions, near infinity.

\subsection{The structure of $L$ at the horizon and infinity \label{waveop}}

To see the structure of $L$ near the horizon, it will be convenient to introduce a set of coordinates. Let us first construct gaussian normal coordinates for the induced metric $h$ such that $\rho=0$ is the horizon and $x^A$ are coordinates on the surfaces $\rho=\textrm{const}.$ generated by pushing $\rho=0$ forward along geodesics normal to the horizon. We use here that the horizon is compact. In these coordinates we have for $0 \leq \rho < \epsilon$
\be
h^{\rho \rho} = 1, \quad h^{\rho A}=0, \quad h^{AB} = \sigma^{AB}(\rho, x^A) \, .
\ee
We also know that (use (\ref{sgrav2}))
\be
W^\rho = \sqrt{A}, \quad W^A=0, \quad \textrm{ on } \rho=0.
\ee
The tensor $a^{ij}$ has components
\be
a^{\rho \rho} = A - (W^\rho)^2, \qquad a^{\rho A} = - W^\rho W^A, \qquad a^{AB} = A \sigma^{AB}-W^A W^B.
\ee
Thus, in this set of coordinates, near the horizon we have the following expansion for the functions occurring in $L$:
\ben{nrhrn}
\begin{array}{rclcrcl}
A &=& A_0 + \O{\rho}, & \qquad & h &=& h_0 + \O{\rho},\\
a^{\rho \rho} &=& C \rho + \O{\rho^2},& \qquad&  a^{A B} &=& A_0 \sigma_0^{A B}+ \O{\rho}, \\
&&& a^{\rho A} = \O{\rho},&&&
\end{array}
\een
where $A_0, h_0, C$ are positive and ${\sigma}_0^{AB}$ is uniformly positive definite on the horizon.

Near infinity we introduce the coordinate $s=\frac{l}{r}$, so that $s=0$ is the conformal boundary. Making use of the AdS asymptotics assumed, we find that the functions appearing in $L$ have the following expansion near $\scri$:
\ben{nrinf}
\begin{array}{rclcrcl}
A &=& s^{-2} \tilde A_0 + \O{1}, & \qquad & h &=& s^{-6} \tilde h_0 + \O{s^{-4}},\\
 a^{ss} &=& \tilde A_0 +\O{s^2}, & \qquad&  a^{AB} &=& \tilde A_0 \tilde\sigma_0^{AB}+ \O{s^2}, \\
&&& a^{s A} = \O{s^2},&&&
\end{array}
\een
where $\tilde A_0, \tilde h_0$ are positive and $\tilde{\sigma}_0^{AB}$ is uniformly positive definite on the conformal boundary.

We sum up the conclusions of this and the previous section in the following Lemma:
\begin{Lemma}  \label{lem41}
\begin{enumerate}[i)]
\item Let $(\mathcal{R}, \mathcal{H}^+, \Sigma, g, r,T)$ be a stationary, asymptotically anti-de Sitter, black hole space time of AdS radius $l$. Then the Klein-Gordon equation
\be
\Box_g \psi + \frac{\alpha}{l^2} = 0,
\ee
can be decomposed as
\be
TT\psi - B T\psi + L \psi = 0,
\ee
where $B$ and $L$ are operators on $\Sigma$ given by \eq{Bdef}, \eq{Ldef}. 

\item A neighbourhood of $\mathcal{H}^+\cap \Sigma$ can be covered with coordinate patches $x^i = (\tilde{\rho}, x^A)$ such that the functions appearing in $L$ locally take the form \eq{nrhrn}.

\item A neighbourhood of $\scri \cap \Sigma$ can be covered with coordinate patches $x^i = (s, x^A)$ such that $\scri = \{s=0\}$ and such that the functions appearing in $L$ locally take the form \eq{nrinf}.
\end{enumerate}
\end{Lemma}

\subsection{Boundedness of solutions to the Klein-Gordon equation}

We will now interest ourselves in the eigenvalues of $L$, i.e. functions $u$ which satisfy suitable boundary conditions (regularity at the horizon, Dirichlet, Neumann or Robin at infinity) together with the equation
\be
Lu = \omega u \, .
\ee
The natural $L^2$ space associated with this operator  has the norm
\ben{Ll2def}
\norm{u}{\Ll^2(\Sigma)}^2 = \int_{\Sigma} u^2  \sqrt{\frac{1}{A}}d S_{\Sigma} \, .
\een
Let us pick a smooth, positive twisting function $g: int (\Sigma) \to \mathbb{R}$. In a neighbourhood
of the horizon we require $g = \abs{\log \rho}$, for $\rho$ as in \S\ref{waveop}, while in a neighbourhood of infinity, we require $g = r^{-\frac{3}{2}+\kappa}$. Associated with the twisting function we define the norm:
\ben{Hhnorm}
\norm{u}{\Hh^1(\Sigma, \kappa)}^2 = \int_\Sigma \left(a^{ij} \tilde{\nabla}_i u \tilde{\nabla}_j u + u^2\right ) \sqrt{\frac{1}{A}} dS_\Sigma,
\een
where the derivatives are twisted by $g$. We define $\Hh^1(\Sigma, \kappa)$ to be the space of measurable functions whose derivatives exist in a weak sense and for which this norm is finite. The space $\dot\Hh^1(\Sigma, \kappa)$ is the completion of the set of smooth functions which vanish in a neighbourhood of the horizon, and the space $\dot\Hh^1_0(\Sigma, \kappa)$ is the completion of $C_c^\infty(\Sigma)$ in the norm \eq{Hhnorm}.

After these preliminaries, we note that we can write our operator $L$ as
\ben{gebu}
L u = \tilde{\nabla}_i^\dagger (a^{ij} \tilde{\nabla}_j u) + V u,
\een
where $\tilde{\nabla}_i^\dagger$ is the formal adjoint of $\tilde{\nabla}_i$ with respect to the $\Ll^2(\Sigma)$ inner product. 
We immediately observe that due to our choice of twisting function Lemma \ref{lem41} implies that
\ben{hebu}
V := \frac{L g}{g},
\een
is smooth in the interior of $\Sigma$ and bounded both at the horizon and at infinity.
Finally, we define the bilinear form associated with (\ref{gebu}) to be
\ben{bilin}
B[u, v] := \int_\Sigma \left(a^{ij} \tilde{\nabla}_i u \tilde{\nabla}_j v + V uv\right ) \sqrt{\frac{1}{A}} dS_\Sigma+ \int_{\scri \cap \Sigma} r^{3-2 \kappa} u v\ \beta d \eta.
\een
Here $d\eta$ is a measure on the conformal boundary defined by:
\be
d\eta = \lim_{r \to \infty} \left(\frac{\sqrt{A}}{l r^3} dS_{\mathcal{K}_r}\right),
\ee
where $dS_{\mathcal{K}_r}$ is the measure induced on the surface $r = \textrm{const.}$ by the metric $h$ on $\Sigma$.  One can check that this gives a finite measure on the conformal boundary. The limit of $r^{3-2\kappa} u v $ as $r \to \infty$ is integrable over $\scri \cap \Sigma$ with this measure, provided $u, v \in \Hh^1(\Sigma, \kappa)$.\footnote{We abuse notation slightly by inserting the limit directly into the integral at infinity.} For Dirichlet or Neumann boundary conditions, $\beta$ is understood to vanish. We can, in the usual way, define weak solutions etc.\ making use of this inner product.

The key insight, formulated in Proposition \ref{prop1} below, is that the standard theory of eigenvalues of a self-adjoint elliptic operator on a finite domain can be extended to the operator $L$. We impose boundary conditions of regularity at the horizon (equivalent in the weak formulation to membership of $\dot\Hh^1(\Sigma, \kappa)$) and either Dirichlet, Neumann or Robin boundary conditions at infinity. The key ingredient for Proposition \ref{prop1} is a generalisation of the Rellich-Kondrachov theorem to the twisted Sobolev spaces:
\begin{Theorem}\label{RKthm}
\begin{enumerate}[i)]
\item The space $\dot\Hh^1(\Sigma, \kappa)$ embeds compactly into $\Ll^2(\Sigma)$ for $0 <\kappa <1$.
\item The space $\dot\Hh^1_0(\Sigma, \kappa)$ embeds compactly into $\Ll^2(\Sigma)$ for $\kappa>0$.
\end{enumerate}
\end{Theorem}
\begin{proof}
See \S \ref{RKthmproof}.
\end{proof}
\begin{Proposition} \label{prop1} For the operator $L$, with either Dirichlet  (for $\alpha <9/4$), Neumann or Robin (for $5/4<\alpha<9/4$) boundary conditions the eigenvalues $\omega$ and their associated eigenfunctions $u$ satisfy:
\begin{enumerate}[(i)]
\item The eigenvalues are real and bounded below.
\item The eigenvalues form a countable sequence $\omega_1\leq \omega_2 \leq \ldots$ which is discrete.
\item Each $\omega_i$ has finite multiplicity.
\item $\omega_i\to \infty$ as $i \to \infty$.
\item \label{evbasis} The eigenfunctions $u_n$ form an orthonormal basis for  $\Ll^2(\Sigma)$.
\item The lowest eigenvalue, $\omega_1$,  can be expressed by the following Rayleigh-Ritz type formulae:
\be
\omega_1 = \min_{u \in \dot\Hh^1_0(\Sigma, \kappa), \norm{u}{\L^2(\Sigma)}=1} B[u,u],
\ee
for the Dirichlet case, and
\be
\omega_1 = \min_{u \in \dot\Hh^1(\Sigma, \kappa), \norm{u}{\L^2(\Sigma)}=1} B[u,u],
\ee
for the Neumann or Robin case.

\item $\omega_1$ has multiplicity one.
\item $u_1$ vanishes only at infinity and is smooth on $\Sigma$, with all derivatives bounded up to the horizon.
\item Given $\epsilon>0$, there exists a $\delta>0$ such that if $B'$ is a symmetric bilinear form satisfying
\be
\abs{B[u,u]-B'[u,u]}<\delta \norm{u}{\Hh^1(\Sigma, \kappa)}^2,
\ee
for all $u \in \dot\Hh^1(\Sigma, \kappa)$, then\footnote{For Dirichlet boundary conditions, replace  $\dot\Hh^1(\Sigma, \kappa)$ with $\dot\Hh^1_0(\Sigma, \kappa)$} 
\be
\omega'_1 = \min_{u \in \dot\Hh^1(\Sigma, \kappa), \norm{u}{\L^2(\Sigma)}=1} B'[u,u],
\ee
exists  and satisfies
\be
\abs{\omega_1-\omega_1'}<\epsilon.
\ee
In particular, this implies that for the Kerr black hole, $\omega_1$ depends continuously on $m$, $l$, $a$ and, if relevant, the (time independent) Robin function, understood as an element of $C^0(\scri)$.
\end{enumerate}
\end{Proposition}
\begin{proof}
The proof of $(i)$-$(vi)$ follows along exactly as in \cite[Chap. 6]{Evans} or \cite[Chap. II]{Lady}. An application of the strong maximum principle, together with an elliptic regularity result as in \cite[\S5.2]{Warnick:2012fi} gives $(vii),(viii)$. The final part follows from the fact that $\omega_1$ can be expressed as the minimum of a bilinear form.
\end{proof}

We remark that in the case that $\omega_1$ has a sign, this sign is independent of the spacelike slicing with respect to which we decompose the wave operator to extract the elliptic operator $L$: 

\begin{Corollary}\label{cor42}
Let $\Sigma'$ be a surface to the future of $\Sigma$, which is smoothly homotopic through spacelike surfaces to $\Sigma$ and such that $(\overline{J^+(\Sigma')}\cap\mathcal{R}, \overline{J^+(\Sigma')}\cap \mathcal{H}^+, \Sigma', g, r, T)$ is a stationary, asymptotically anti-de Sitter black hole spacetime. Let $\{\Sigma'_t:= \varphi_t(\Sigma')\}$ be the associated spacelike slicing of $\overline{J^+(\Sigma')}\cap\mathcal{R}$. This slicing has an associated elliptic operator $L'$. Then the smallest eigenvalue of $L'$, $\omega_1'$, must have the same sign as $\omega_1$, the smallest eigenvalue of $L$. 
\end{Corollary}
\begin{proof}
It suffices to exclude the case where $\omega_1>0$ and $\omega'_1<0$. By the continuity of the eigenvalues, there must exist an intermediate surface, say $\Sigma''$, for which the associated $\omega''_1=0$. Thus there exists a non-trivial static solution to the Klein-Gordon equation on $\mathcal{R}$. This implies that for any slicing, the corresponding elliptic operator must have a zero eigenvalue which is in contradiction with the assumption that $\omega_1>0$.
\end{proof}

We now claim that boundedness of solutions to the massive wave equation \eq{vas} is  controlled by the lowest eigenvalue of the operator $L$. First, let us make use of the fact that $u_1$ is a non-vanishing eigenfunction of $L$ to simplify \eq{weq} by twisting with $u_1$. We may re-write the operators $B$, $L$ as
\be
B u = {W^i}  u_1 \frac{\partial }{\partial x^i} \frac{u}{u_1}+\sqrt{\frac{A}{ h}} \frac{1}{u_1}\frac{\partial}{\partial x^i}\left( \sqrt{\frac{ h}{A}} W^i u_1 u\right),
\ee
and
\be
L u = - \sqrt{\frac{A}{ h}} \frac{1}{u_1} \frac{\partial}{\partial x^i}\left( u_1^2 \sqrt{\frac{ h}{A}} a^{ij} \frac{\partial  }{\partial x^j}\frac{u}{u_1}\right) + \omega_1 u.
\ee
Note that due to the twisting, we were able to remove the term in the elliptic part of the operator which is proportional to $A u$ and replaced it with the term $\omega_1 u$. Since the AdS asymptotics imply that $A$ grows like $r^2$ near infinity this gains us two powers of $r$. This is important when we are in the range of Neumann / Robin boundary conditions, as it makes the energy finite for the slower fall-off. 

We can construct the relevant energy currents by any of the methods above, or alternatively by multiplying \eq{weq} by $\psi_t \sqrt{h/A}$ and integrating over $\Sigma$ from the horizon out to $\scri$. We  define the energy to be 
\be
E(t)=  \frac{1}{2}\int_{\Sigma_t} \left[\left (\nabla_t \psi \right)^2+a^{ij} \tilde \nabla_i \psi   \tilde \nabla_j \psi +\omega_1 \psi^2 \right] \sqrt{\frac{1}{A}}d S_{\Sigma_t} \, ,
\ee
where now the twisted derivatives are understood to be twisted by $u_1$. Note that we do not include a surface term here, even for Robin boundary conditions. The reason for this is that twisting by a  function obeying the Robin boundary conditions modifies the boundary conditions to simpler Neumann boundary conditions.
The flux across the horizon is given by
\be
 F[T_1, T_2]= \int_{\mathcal{H}^+_{[T_1, T_2]}}   \left(\partial_t \psi \right)^2 W^i n_i \sqrt{\frac{1}{A}}   d S_{\mathcal{H}^+},
\ee
where $\mathcal{H}^+_{[T_1, T_2]}$ is the portion of $\mathcal{H}^+$ lying between $\Sigma_{T_1}$ and $\Sigma_{T_2}$ and $n$ is the outward\footnote{i.e. pointing into $\mathcal{R}$} unit normal of $\mathcal{H}^+\cap \Sigma_t$ as a surface embedded in $\Sigma_t$. This flux is positive by virtue of the positivity of the surface gravity (recall \eq{sgrav2}). We have, of course, the energy identity
\be
E[T_2] = E[T_1] - F[T_1, T_2].
\ee

Now, clearly after making use of the standard redshift arguments \cite{Mihalisnotes} to resolve the degeneration of the energy at the horizon, we have the following, which is the main result of this paper:
\begin{Theorem}\label{mainthm}
Let $(\mathcal{R}, \mathcal{H}^+, \Sigma, g, r,T)$ be a stationary, asymptotically anti-de Sitter, black hole space time of AdS radius $l$. Fix Dirichlet, Neumann or Robin boundary conditions. Suppose that $\omega_1>0$ for these boundary conditions. Then $E[t]$ is positive definite. Furthermore, the conclusions of Theorem  \ref{schbdn} hold as stated, with $\mathcal{R}, \mathcal{H}^+, \Sigma, g, r$ understood to refer to the more general stationary, aAdS, black hole.
\end{Theorem}

For $\omega_1<0$ we do not immediately get any boundedness result. In fact -- provided $\omega_1$ is the only negative eigenvalue of $L$ -- we find that there are solutions which grow in time. 

\begin{Theorem}
Suppose $\omega_1 < 0$ and $\omega_2 >0$, then given $J\in \mathbb{N}$ there exists a solution of \eq{vas} for which $\norm{\phi}{\Ll^2(\Sigma_t)}$ grows at least as fast as $t^J$.
\end{Theorem}
\begin{proof}
First, we note that by property ({\it\ref{evbasis}}) of Proposition \ref{prop1}, we may expand any smooth solution $\phi(t)$ as
\ben{expand}
\phi(t) = \sum_{n = 1}^\infty f_n(t) u_n,
\een
where furthermore, for any $i$, $j$ we have that the sum
\be
\sum_{n=1}^\infty \abs{\omega_n}^i \left[\frac{d^j f_n}{dt^j}(t)\right]^2
\ee
converges uniformly\footnote{we actually only require this for $i=0, 1$} in $t$ on any interval $[T_1, T_2]$. Now, inserting the expansion \eq{expand} into the energy identity, we deduce that
\be
\sum_{n=1}^\infty \left( \dot{f}_n(t)^2 + \omega_n f_n(t)^2 \right) \leq \sum_{n=1}^\infty \left( \dot{f}_n(0)^2 + \omega_n f_n(0)^2 \right).
\ee
Recall that $\omega_1<0$ and $\omega_n > 0$ for $n>1$. We therefore deduce that if the right hand side is initially negative, say
\be
 \sum_{n=1}^\infty \left( \dot{f}_n(0)^2 + \omega_n f_n(0)^2 \right) = \omega_1<0
\ee
and also $f_1(0)>0$, then we must have
\begin{equation*}
f_1(t) \geq 1,
\end{equation*}
since the term proportional to $f_1(t)^2$ is the only one contributing a negative sign on the left hand side ($f_1$ cannot change sign as it is bounded away from $0$). We can clearly arrange this situation by taking $\phi$ to be the unique solution to \eq{vas} with initial conditions: 
\begin{equation}\label{phi0}
{ \phi  |_{\Sigma_0} = u_1}, \qquad { T\phi |_{\Sigma_0} = 0}.
\end{equation}

We will prove inductively that for each $J$ there exists a solution of \eq{vas}, $\phi^{(J)} =\sum f^J_n(t) u_n$, with initial data in $C^\infty(\Sigma\setminus\scri)$ for which
\begin{equation}\label{fexp}
\left( \frac{d}{dt}\right)^{J} f^J_1(t) \geq 1,
\end{equation}
for all time. This suffices to prove the theorem, since by integrating in time, we conclude that $f^J_1(t)$ grows like $t^J$ at late times. Since $\norm{\phi^{(J)}}{\Ll^2(\Sigma_t)}^2=\sum_n \abs{f^J_n}^2$ the result follows.

As noted above, we can find a $\phi^{(0)}$ satisfying \eq{fexp} for $J=0$, by taking $\phi^{(0)}$ to be the solution with initial data \eq{phi0}. Now suppose for induction that $\phi^{(J)}$ satisfies \eq{fexp} for some $J \geq 0$. We would like to define $\phi^{(J+1)}$ to be the unique solution of \eq{vas} with initial conditions
\begin{equation*}
{ \phi^{(J+1)}  |_{\Sigma_0}} = {L^{-1} \left(B  \phi^{(J)} |_{\Sigma_0} -  T \phi^{(J)}  |_{\Sigma_0} \right)}, \qquad { T\phi^{(J+1)}  |_{\Sigma_0}} = {\phi^{(J)}  |_{\Sigma_0}},
\end{equation*}
but we first need to verify that these are valid initial conditions for our well posedness theorem. Note that $\phi^{(J)}  |_{\Sigma_0}\in \H^1(\Sigma, \kappa)\cap C^\infty(\Sigma\setminus\scri)$ and $T\phi^{(J)} |_{\Sigma_0} \in L^2(\Sigma)\cap C^\infty(\Sigma\setminus\scri)$, so that the bracket on which $L^{-1}$ acts belongs to $L^2(\Sigma)= \mathscr{L}^2(\Sigma)$. The conditions imposed on the eigenvalues of $L$ ensure that $L^{-1}$ exists and maps $\mathscr{L}^2(\Sigma)$ into $\dot{\mathscr{H}}^1(\Sigma, \kappa)$. An elliptic regularity result gives that\footnote{$\dot{\mathscr{H}}^1(\Sigma, \kappa)$ and $ \H^1(\Sigma, \kappa)$ differ only in the degree of differentiability assumed at the horizon.} $\phi^{(J+1)}  |_{\Sigma_0}\in \dot{\mathscr{H}}^1(\Sigma, \kappa)\cap C^\infty(\Sigma\setminus\scri) = \H^1(\Sigma, \kappa)\cap C^\infty(\Sigma\setminus\scri)$. Since furthermore $T\phi^{(J+1)} |_{\Sigma_0} \in L^2(\Sigma)\cap C^\infty(\Sigma\setminus\scri)$, these initial conditions do indeed launch a solution.

Now, by construction we have that $T\phi^{(J+1)} = \phi^{(J)}$ for all time. To see this, observe that $T\phi^{(J+1)} - \phi^{(J)}$ solves \eq{vas} (or equivalently \eq{weq}), with trivial initial conditions  so vanishes everywhere. Thus $\frac{d}{dt} f^{J+1}_i = f^{J}_i$, whence we have established \eq{fexp} holds.
\end{proof}
Finally we consider the case that $\omega_1=0$. In this case, we have
\begin{Theorem}
Suppose $\omega_1=0$. Then if $\psi$ solves \eq{vas}, $\norm{\psi}{\Ll(\Sigma_t)}$ can grow at most linearly in $t$.
\end{Theorem}
\begin{proof}
Twisting by $u_1$, as for the proof of Theorem \ref{mainthm}, we find an energy which is positive, but which does not control $\norm{u}{\Hh^1(\Sigma)}$. It does however control $\norm{\partial_t \psi}{\Ll^2(\Sigma)}$, so that
\be
\norm{\partial_t \psi}{\Ll^2(\Sigma_t)}^2 \leq E[0].
\ee
Now consider
\bean
\frac{d}{dt}\left (\norm{\psi}{\Ll^2(\Sigma_t)}^2\right) &=& 2 \norm{\psi}{\Ll^2(\Sigma_t)}\frac{d}{dt}\left(\norm{\psi}{\Ll^2(\Sigma_t)}\right)  \\
&=& 2 \ip{\psi}{\partial_t\psi}{\Ll^2(\Sigma_t)} \\ &\leq& 2 \norm{\psi}{\Ll^2(\Sigma_t)}\norm{\partial_t \psi}{\Ll^2(\Sigma_t)}
\eean
whence we deduce that for almost every $t$
\be
\frac{d}{dt}\left(\norm{\psi}{\Ll^2(\Sigma_t)}\right) \leq \norm{\partial_t \psi}{\Ll^2(\Sigma_t)} \leq \sqrt{E[0]}
\ee
and the result follows.
\end{proof}

A few final comments are in order. Firstly, we note that the above theorems justify the assertion (cf. \cite{Dias:2010ma}) that a spacetime which admits linear scalar hair (i.e.\ a non-trivial stationary solution) is at the threshold between stability and instability for the Klein-Gordon equation. See also the remark below Theorem \ref{thm51}.

Note also that our argument does not depend on the specific form of the initial twisting function $g$. The correct asymptotics near the horizon and near infinity for $g$ were sufficient to generate a bounded potential term in the twisted equation (cf.~(\ref{gebu}) and (\ref{hebu})) and gave as an abstract conclusion the existence of a lowest eigenvalue $\omega_1$ for $L$ with its associated eigenfunction $u_1$. In  a second step, we twisted with that eigenfunction to obtain either a boundedness or instability result.

In practice it may be hard to compute the lowest eigenvalue and eigenfunction explicitly. However, to prove stability, one can try to find a function $g$ with the correct asymptotics such that $\frac{Lg}{g}$ is at least equal to some non-negative function. In fact, this was precisely what we did in Section \ref{sec:SchwSch}.

Finally, we note that our approach here is based purely on the spectral properties of $L$. For a full understanding of the global behaviour of solutions to \eq{vas} one should study the spectral properties of the full wave operator, involving both $L$ and $B$. In this way one is led to the consideration of quasinormal modes, see \cite{claudeQNM}.

\section{The AdS-Kerr Black Hole} \label{sec:KerrAdS}

We now apply our general results to the special case of greatest interest, that of the AdS-Kerr black hole. In Boyer-Lindquist coordinates, the metric takes the form
\ben{kerrBL}
g = -\frac{\Delta_-}{\Sigma}\left( d\tilde t+ \frac{a}{\Xi} \sin^2\theta d\tilde \phi\right)^2+\frac{\sin^2\theta \Delta_\theta}{\Sigma}\left(\frac{r^2+a^2}{\Xi} d\tilde \phi - a d\tilde t \right)^2+\frac{\Sigma}{\Delta_-} dr^2 + \frac{\Sigma}{\Delta_\theta} d \theta^2, 
\een
where we have
\be
\begin{array}{rclcrcl}
\Sigma &=& r^2 + a^2 \cos^2\theta,& \qquad & \Delta_\theta &=& 1-\frac{a^2}{l^2} \cos^2 \theta, \\
\Delta_- &=& \left(r^2+a^2 \right) \left(1+\frac{r^2}{l^2} \right)- 2M r,  & \qquad & \Xi &=& 1- \frac{a^2}{l^2}, 
\end{array}
\ee
This metric describes a rotating black hole in a background with an asymptotically anti-de Sitter end of radius $l$ provided $\abs{a/l}<1$, which we assume henceforth. The coordinate $r$ is a good asymptotic radial coordinate. We refer to $M$ as the mass of the black hole and $a$ as the rotation parameter. The metric has a Killing horizon located at $r_{+}$, defined to be the largest root of $\Delta_-$. The Hawking-Reall Killing vector
\be
T = \frac{\partial}{\partial \tilde t} + \lambda \frac{\partial}{\partial \tilde \phi}, \qquad \lambda = \frac{a \Xi}{r_{+}^2 + a^2},
\ee
is null on the horizon. Provided $\abs{a}l < r_{+}^2$, Hawking and Reall  \cite{Hawking:1999dp} observed that $T$ is in fact timelike everywhere outside the horizon. We henceforth assume $\abs{a}l < r_{+}^2$ also.

Boyer-Lindquist coordinates have the advantage that the wave equation separates, owing to the existence of a hidden constant of the motion due to Carter. They suffer from the disadvantage of not being regular at the horizon. Let us make the following coordinate transformation
\ben{tran}
t=\tilde t+A(r) \, , \quad \textrm{ and } \quad \phi= \tilde \phi + B(r) -\lambda [\tilde t+ A(r)],
\een
where
\be
\frac{dA}{dr} = \frac{2 M r}{\Delta_-\left( 1+\frac{r^2}{l^2}\right)}, \quad \textrm{ and } \quad \frac{dB}{dr} = \frac{a \Xi}{\Delta_-}.
\ee
In these coordinates we have $T = \partial_t$. A calculation verifies that
\begin{Lemma}
Let $\mathcal{R} = [0, \infty)_t \times [r_{+}, \infty)_r \times S^2_{\theta, \phi}$ endowed with the metric $g$ resulting from applying the coordinate transformation \eq{tran} to \eq{kerrBL}. Assume $|a|<l$ and $\abs{a}l < r_{+}^2$. Take $\mathcal{H}^+ = \{ r = r_{+}\}$, $\Sigma = \{ t = 0\}$. Then $(\mathcal{R}, \mathcal{H}^+, \Sigma, g, r, T)$ is a stationary, asymptotically anti-de Sitter black hole spacetime with AdS radius $l$ in the sense of Definition \ref{adsbh}.
\end{Lemma}
Thus we may apply all of the results of the previous section to the AdS-Kerr black hole. The transformation \eq{tran} puts the metric directly into the form \eq{statg}, whence we may directly read off $A$, $W$ and $h$ and construct the operator $L$ whose eigenvalues control the boundedness of solutions to the wave equation. For general boundary conditions, this is an operator on the three-dimensional space with coordinates $r, \theta, \phi$, and is somewhat ugly. Matters simplify when the boundary conditions are consistent with the axial symmetry of the black hole. This occurs for Dirichlet or Neumann boundary conditions, as well as for Robin boundary conditions where $\beta = \beta(\theta)$ is axisymmetric on the sphere at infinity. In this case, an averaging argument shows that the least eigenvalue of $L$ corresponds to an axially symmetric eigenfunction. We have the following result:
\begin{Theorem}\label{thm51}
Fix an AdS-Kerr black hole background obeying the Hawking-Reall bound, $\alpha<\frac{9}{4}$, and Dirichlet, Neumann or axisymmetric Robin boundary conditions, as appropriate for the choice of $\alpha$. Let $\omega_1$ be the least eigenvalue of the eigenvalue problem\footnote{For convenience, we state the eigenvalue problem in untwisted form. Again, the existence of a lowest eigenvalue for this problem follows from the general arguments of Section \ref{sec:genstat}.}:
\ben{Kerrop}
 - \left[ \frac{1}{\Sigma \sin \theta} \frac{\partial}{\partial r} \left(\Delta_- \frac{\partial u}{\partial r}  \right)+\frac{1}{\Sigma \sin \theta} \frac{\partial}{\partial \theta} \left(\Delta_\theta \frac{\partial u}{\partial \theta}  \right)+ \frac{\alpha}{l^2} u \right] = \omega \abs{g^{tt}} u,
\een
subject to the conditions that $u$ be regular at $\theta = 0, \pi$ and $r = r_{+}$ and that near infinity we have
\be
u \sim \frac{u_0}{r^{\frac{3}{2}-\kappa}} + \frac{u_1}{r^{\frac{3}{2}+\kappa}},
\ee
where $\kappa = \sqrt{9/4-\alpha}$ and:
\be
\begin{array}{ccl}
u_0  =0&\qquad& \textrm{for Dirichlet boundary conditions,} \\ 
u_1 =0&\qquad& \textrm{for Neumann boundary conditions,} \\
 2\kappa u_1 -\beta(\theta) u_0=0  &\qquad& \textrm{for Robin boundary conditions.}
 \end{array}
 \ee
 Then:
\begin{enumerate}[i)]
\item If $\omega_1>0$, Theorem \ref{schbdn} holds for solutions $\psi$ of the Klein-Gordon equation \eq{vas} on this background, satisfying the given boundary conditions. That is, solutions are bounded pointwise in time.
\item  If $\omega_1<0, \omega_2>0$, there exist solutions $\psi$ of the Klein-Gordon equation \eq{vas} on this background, satisfying the given boundary conditions, whose energy grows faster than any power of $t$.
\end{enumerate}
\end{Theorem}

\textbf{Remark. }Recall that as we vary $m, l, a$ and $\beta$ smoothly, the corresponding $\omega_1$ varies continuously. In order to pass from a spacetime in which solutions to the Klein-Gordon equation are bounded to one in which they grow without bound, $\omega_1$ must pass through $0$. For this specific set of parameters, the Klein-Gordon equation will admit a non-trivial stationary solution, i.e.\ linear scalar hair. As in the case of Schwarzschild (c.f.\ Figure \ref{betfig}), if possibility $i)$ holds for the Neumann boundary conditions\footnote{See Theorem \ref{theo:2}.}, we can always induce a transition to possibility $ii)$ by taking $\beta$ increasingly large and negative.

As for Schwarzschild, we can replace the weight $\abs{g^{tt}}$ appearing on the right hand side of \eq{Kerrop} by any smooth, positive, function with the same asymptotic behaviour. For example $r^{-2}$ will do. Notice that the operator appearing on the left hand side of \eq{Kerrop} is the wave operator in Boyer-Linquist coordinates acting on a stationary axisymmetric field.

\subsection{The Dirichlet case}

We shall now demonstrate that for Dirichlet conditions, $\omega_1$ is always positive, provided the black hole satisfies the Hawking-Reall bound. To do this, we multiply \eq{Kerrop} by $u \Sigma \sin \theta$ and integrate over $r, \theta$. After integrating by parts (which we may do for $u$ satisfying Dirichlet conditions), it suffices to show that
\[
Q := \int_{r_{+}}^\infty \int_0^\pi \Big[  \Delta_- \left(\partial_r u\right)^2 + \Delta_\theta \left(\partial_\theta u \right)^2 - \frac{\alpha}{l^2} \left(r^2 + a^2 \cos^2 \theta\right) u^2 \Big] \sin \theta d\theta dr \geq 0 \, ,
\]
with equality only for $u \equiv 0$.
\begin{Proposition}
For any $\alpha \leq \frac{9}{4}$ we have $Q\geq 0$, with equality if and only if $u \equiv 0$.
\end{Proposition}
\begin{proof}
Clearly, it suffices to prove the statement for $\alpha = \frac{9}{4}$. The Proposition will be an immediate consequence of two Lemmas established below: 
Adding the estimates (\ref{maina}) (integrated in $\int_{r_+}^\infty dr$) and the estimate (\ref{mainr}) (integrated in $\int_0^\pi \sin \theta d\theta$) yields $Q >0$ for $u \in C^1$, unless $u = 0$. We can then pass to a limit of continuous functions to establish the result for functions in $\Hh^1$.
\end{proof}

\begin{Lemma}
For any function $u \in C^1([0, \pi])$, we have the estimate
\begin{align} \label{maina}
\frac{9}{4} \frac{a^2}{l^2}  \int_0^\pi u^2 \cos^2 \theta \sin \theta d\theta \leq \int_0^\pi \Delta_\theta \left(\partial_\theta u \right)^2 \sin \theta d\theta +\frac{3}{2} \frac{a^2}{l^2} \int_0^\pi u^2 \sin \theta d\theta,
\end{align}
with strict inequality unless $u \equiv 0$.
\end{Lemma}
\begin{proof}
Consider the inequality
\begin{align}
\int_0^\pi  \left(\sin \theta \cdot \partial_\theta u + \gamma \cos \theta u\right)^2 \ \sin \theta d\theta \geq 0,
\end{align}
for any constant $0<\gamma<3$. Squaring, integrating the mixed term by parts and using $\cos^2 \theta + \sin^2 \theta = 1$ yields the inequality
\begin{align} 
\frac{a^2}{l^2} \int_0^\pi \sin^3 \theta \left(\partial_\theta u \right)^2 d\theta + \frac{a^2}{l^2} \gamma \int_0^\pi u^2 \sin \theta d\theta \geq \frac{a^2}{l^2} \left(3\gamma - \gamma^2 \right) \int_0^\pi u^2 \cos^2 \theta \sin \theta d\theta \, . \nonumber
\end{align}
Noting that $\Delta_\theta = \frac{a^2}{l^2} \sin^2 \theta + \Xi$  we consequently have
\begin{align} 
\int_0^\pi \Delta_\theta \left(\partial_\theta u \right)^2 \sin \theta d\theta + \frac{a^2}{l^2} \gamma \int_0^\pi u^2 \sin \theta d\theta \geq \frac{a^2}{l^2} \left(3\gamma - \gamma^2 \right) \int_0^\pi u^2 \cos^2 \theta \sin \theta d\theta . \nonumber
\end{align}
Choosing $\gamma = \frac{3}{2}$ yields (\ref{maina}).
This inequality is indeed strict unless $u \equiv0$: Since $\Xi>0$, this is immediate provided $u$ is not $\theta$-independent. However, it is strict also in the latter case as can be checked by explicit integration.
\end{proof}

Now, let us treat the radial part
\begin{Lemma}
For any $u \in C^1\left[r_+,\infty\right)$ satisfying $u r^\frac{3}{2} \rightarrow 0$ as $r \rightarrow \infty$ we have
\begin{align} \label{mainr}
\frac{1}{l^2} \int_{r_{+}}^\infty \left( \frac{9}{4} r^2 + \frac{7}{4}a^2 \right) u^2 dr  \leq \int_{r_{+}}^\infty  \Delta_- \left( \partial_r u  \right)^2 dr  \, .
\end{align}
\end{Lemma}
\begin{proof} 
Let $f: \left[r_{+},\infty\right) \rightarrow \mathbb{R}$ be a continuous function which is differentiable in $\left(r_{+},\infty\right)$ and  such that $f r^{-1}$ is uniformly bounded. We have the identity
\begin{align}
\int_{r_{+}}^\infty \left(\sqrt{\Delta_-}  \partial_r u + f u \right)^2 dr \geq 0 \, ,
\end{align}
which after integration by parts may be rewritten as
\begin{align} \label{aim0}
\int_{r_{+}}^\infty \left(-f^2 + \partial_r \left(\sqrt{\Delta_-} f \right) \right) u^2 \ dr \leq \int_{r_{+}}^\infty  \Delta_- \left( \partial_r u  \right)^2 dr \, .
\end{align}
We claim there exists an admissible $f(r)$ such that
\be
-f(r)^2 + \frac{d}{dr} \left(\sqrt{\Delta_-} f(r) \right) \geq \frac{9}{4} \frac{r^2}{l^2} + \frac{7}{4} \frac{a^2}{l^2}.
\ee
We note that we can factorise
\bean
\Delta_- &=& \frac{1}{l^2} (r-r_{+})\left(r^3 + r^2 r_{+} + r(r_{+}^2+l^2+a^2)-\frac{a^2 l^2}{r_{+}}\right ), \\
&=& \frac{1}{l^2}(r-r_{+}) h(r).
\eean
Let us write
\be
f(r) = \frac{3}{2 l} \sqrt{\frac{r-r_{+}}{h(r)}} g(r),
\ee
where $g(r)$ is a quadratic function in $r$ whose $r^2$ coefficient is unity. Clearly it suffices to prove that we may choose $g$ such that
\be
\frac{-9}{4l^2} \frac{ (r-r_{+}) g(r)^2}{h(r)} + \left[ \frac{3}{2 l^2} \frac{d}{dr} \left( (r-r_{+}) g(r) \right)-\frac{9}{4 l^2} r^2 -\frac{7}{4 l^2} a^2 \right] \geq 0.
\ee
Let us choose $g(r)$ such that the term inside the square bracket is
\be
 \left[ \frac{3}{2 l^2} \frac{d}{dr} \left( (r-r_{+}) g(r) \right)-\frac{9}{4 l^2} r^2 -\frac{7}{4 l^2} a^2 \right]  = \frac{9}{4 l^2} (r-r_{+})(r+r_{+}).
\ee
A brief calculation shows that for this we should take
\be
g(r) = r^2 + r r_{+} + \frac{1}{6}\left( 7 a^2 - 3 r_{+}^2\right).
\ee
Noting that $\frac{(r-r_{+})}{h(r)} \geq 0$, it remains then to prove that
\be
F(r):=-g(r)^2 + (r+r_{+}) h(r) \geq 0.
\ee
Thanks to our happy choice of $g$, $F(r)$ is a \emph{quadratic} in $r$, so all that remains to us is to verify that it is positive in $r\geq r_{+}$. We calculate that
\bean
F(r) &=& \left[ l^2 \left(1-\frac{a^2}{l^2} \right)+\frac{r_{+}^2}{3} \left(1-\frac{a^2}{r_{+}^2} \right)\right](r^2+r r_{+})  \\ && \quad + \frac{5}{3} r_{+}^2 (r^2+ r r_{+}) -\frac{a^2l^2}{r_{+}} r \\
&& \quad + \frac{1}{36}\left(42 a^2 r_{+}^2 - 9 r_{+}^4 -36 a^2 l^2 -49 a^4 \right).
\eean
Differentiating, we have
\bean
F'(r) &=& \left[ l^2 \left(1-\frac{a^2}{l^2} \right)+\frac{r_{+}^2}{3} \left(1-\frac{a^2}{r_{+}^2} \right)\right](2r+ r_{+})  \\ && \quad + \frac{5}{3} r_{+}^2 (2 r+r_{+}) -\frac{a^2l^2}{r_{+}} .
\eean
Making use of $r\geq r_{+}$ and $r_{+}^2 \geq a l$, we deduce that $F'(r)>0$ for $r>r_{+}$, thus $F(r) \geq F(r_{+})$. Finally then, we calculate
\be
F(r_{+}) = \left(2l^2 r_{+}^2 +\frac{3}{2} r_{+}^4\right)\left(1-\frac{a^2}{r_{+}^2} \right) + \frac{49}{36} r_{+}^4 \left(1-\frac{a^4}{r_{+}^4} \right)+\frac{8}{9} r_{+}^4 > 0,
\ee
which completes the proof of the Lemma. \end{proof}

We sum up then with
\begin{Theorem}\label{thm52}
Smooth solutions $\psi$ to \eq{vas} on the exterior of an AdS-Kerr black hole satisfying $|a|<l$ and the Hawking-Reall bound $r_+^2>|a|l$, subject to Dirichlet conditions at infinity are bounded pointwise in time, up to and including the horizon.
\end{Theorem}

This proves the Dirichlet part of Theorem \ref{theo:2} in the introduction. For the Neumann-part of that theorem, since $|a|<l$ is small, we may argue by continuity using the resolution of the Schwarzschild problem. From Proposition \ref{enlem} we know that for $a=0$ and any $M, l$ and $5/4<\alpha<9/4$ we have $\omega_1>0$, and hence boundedness of solutions to the Neumann problem. For each choice of $M, l, \alpha$, by continuity of $\omega_1$ this result will hold for $\abs{a} < a_c$, where $a_c$ is defined by
\be
a_c = \min \left(\{l\}\cup \{\abs{a} : \abs{a} l = r_+^2\} \cup \{ \abs{a} : \omega_1 = 0\} \right).
\ee
In other words, boundedness holds for solutions to the massive wave equation with Neumann boundary conditions on Kerr-AdS black holes up to the point where either the Hawking-Reall bound is saturated or else linear scalar hair appears. Numerical investigations in \cite{Dias:2010ma}, as well as our own numerical studies suggest that no linear scalar hair appears for black holes obeying the Hawking-Reall bound:
\begin{conj*}
For Neumann boundary conditions, Theorem \ref{theo:2} also holds assuming only $|a|<l$ and $r_+^2 > |a| l$.
\end{conj*}

To establish this rigorously, it would suffice to find a smooth, positive function $u$, obeying the relevant Neumann conditions at infinity such that
\be
 - \left[ \frac{1}{\Sigma \sin \theta} \frac{\partial}{\partial r} \left(\Delta_- \frac{\partial u}{\partial r}  \right)+\frac{1}{\Sigma \sin \theta} \frac{\partial}{\partial \theta} \left(\Delta_\theta \frac{\partial u}{\partial \theta}  \right)+ \frac{\alpha}{l^2} u \right] \geq 0,
\ee
where the left-hand side should not vanish everywhere. While we can find such a function for certain subsets of the parameter space, we have not yet found a $u$ which demonstrates boundedness on every black hole obeying the Hawking-Reall bound and for all values $5/4<\alpha<9/4$.

\section{Proof of Theorem \ref{RKthm}\label{RKthmproof}}

We wish to prove the compactness of the embedding $\Hh^1(\Sigma, \kappa) \hookrightarrow \Ll^2(\Sigma)$. In order to do this, we shall first consider a simpler problem on the half-space $\mathbb{R}^N_+$, and then show how a partition of unity argument can be applied to obtain the full result. 

Let us write $\mathbb{R}^N_+ = \{(x, x^a) \in \mathbb{R}^N: x\geq 0\}$. We assume that $U \subset \mathbb{R}^N_+$ is a bounded Lipschitz domain which may or may not intersect the boundary $x=0$. We define the following function spaces:
\begin{Definition}
\begin{enumerate}[i)]
\item Let $w$ be a real function which is smooth and positive on $(0, \infty)$. A measurable function $u$ belongs to $L^2(U; w)$ provided the norm
\be
\norm{u}{L^2(U;w)}^2 = \int_U u^2 w(x) dx dx^a,
\ee
is finite.
\item Let $f$ be a real function which is smooth and positive on $(0, \infty)$. We also assume $\norm{f}{L^2(U; w)}$ is finite, where we understand $f$ as a function on $U$ to mean $f(x, x^a) = f(x)$. For a differentiable function $u$ we define the $f-$twisted derivative and its adjoint
\be
\df_i u= f \frac{\partial}{\partial x^i} \left( \frac{u}{f}\right), \qquad \dfd_i u= -\frac{1}{w f} \frac{\partial}{\partial x^i} \left( w f u\right).
\ee
\item We say that $u \in H^1(U; w, f)$ if $u \in L^2(U; w)$, $\df_i u$ exists in the weak sense and
\be
\norm{u}{H^1(U; w, f)}^2 = \norm{u}{L^2(U; w)}^2 + \sum_i \norm{\df_i u}{L^2(U; w)}^2 <\infty.
\ee
\item The space  $H^1_0(U; w, f)$ is the completion of $C^\infty_c(U)$ in the $H^1(U; w, f)$ norm.
\end{enumerate}
\end{Definition}
The embedding $H^1(U;w,f) \subset L^2(U;w)$ is obviously continuous, since
\be
\norm{u}{L^2(U; w)} \leq \norm{u}{H^1(U; w, f)},
\ee
for any $u \in H^1(U;w,f)$. In order to establish that this embedding is in fact compact, we need to show that any bounded sequence $\{u_n\}$ in $H^1(U;w,f)$ has a subsequence which converges strongly in $L^2(U;w)$. This will impose conditions on the functions $w, f$. We define the following two properties for the functions $f, w$.
 \begin{Definition} We say that $f, w$ have \textbf{Property A} on the domain $U$ if functions of the form
 \be
 u = f v, \qquad v \in C^\infty(\Ub),
 \ee
 are dense in $H^1(U; w, f)$.
\end{Definition}
 \begin{Definition} We say that $f, w$ have \textbf{Property B} if there exists $\epsilon>0$ such that the function defined for $0\leq x_0<\epsilon$, $0<L<\epsilon$
\be
h(L, x_0) = \left(\int_{x_0}^{L+x_0} \frac{1}{w(t) f(t)^2} dt\right)\left(\int_{x_0}^{L+x_0}{w(t) f(t)^2} dt\right),
\ee
tends to zero uniformly in $x_0$ as $L \to 0$.
\end{Definition}
Note that if $f, w$ have Property A (or B) on a domain $U$, then their restrictions also have Property A (resp.\ B) on any domain $V\subset U$.

We shall consider $U$ to be the half-ball $B_+^\delta =  \{(x, x^a)\in \mathbb{R}^N : x\geq 0, x^2+x^a x^a \leq \delta\}$, and let $\{u_n\}$ be a sequence of functions in $H(B_+^\delta; w, f)$ which vanish in a neighbourhood of the curved surface of $B_+^\delta$. Using the weak compactness of $H^1(B_+^\delta;w,f)$, we may assume without loss of generality that $\{u_n\}$ converges weakly to some $u \in H^1(B_+^\delta; w, f)$.  We shall show that in fact, $\{u_n\}$ converges strongly in $L^2(B_+^\delta; w)$ provided properties $A$ and $B$ hold and that $\delta$ is sufficiently small. Key to establishing this result is a Poincar\'e inequality on cubes of side length $L$

\begin{Lemma}[Twisted Poincar\'e inequality]\label{Poin}
Let $\Pi = \{(x, x^a):  x_0\leq x \leq x_0+L, \ \ 0\leq x^a \leq L\}$ and suppose that $f, w$ satisfy property A  for domain $\Pi$. Then we have the following inequality for $u\in H^1(\Pi; w, f)$.
\ben{comp1}
\norm{u}{L^2(\Pi; w)}^2 \leq \frac{\ip{u}{f}{L^2(\Pi; w)}^2}{\norm{f}{L^2(\Pi; w)}^2} + C(L, x_0) \sum_i \norm{\tilde{\nabla}_i u}{L^2(\Pi; w)}^2,
\een
where
\be
C(L, x_0) = N \max \{ L^2, h(L, x_0) \}.
\ee
\end{Lemma}
\begin{proof}
We first assume that $\frac{u}{f} \in C^\infty(\overline{\Pi})$. We use the fundamental theorem of calculus as follows
\bean
\frac{u(x, x^a)}{f(x)} - \frac{u(y, y^a)}{f(y)} &=& \int_y^x \frac{d}{dt}\left(\frac{ u(t, x^1, \ldots, x^{N-1})}{f(t)} \right) dt + \frac{1}{f(y)} \int_{y^1}^{x^1} \frac{d}{dt}\left(u(y, t,x^2 \ldots, x^{N-1}) \right) dt \\ && \quad +\ldots+ \frac{1}{f(y)}\int_{y^{N-1}}^{x^{N-1}} \frac{d}{dt}\left(u(y, y^1,y^2 \ldots, t) \right) dt.
\eean
Squaring both sides of this equation, we deduce
\bean
 \frac{u(x, x^a)^2}{f(x)^2} + \frac{u(y, y^a)^2}{f(y)^2} - 2 \frac{u(x, x^a) u(y, y^a)}{f(x)f(y)} 
&\leq& N \left( \int_y^x \frac{d}{dt}\left(\frac{u(t, x^1, \ldots, x^{N-1})}{f(t)} \right) dt \right)^2+  \\ && + N\left( \frac{1}{f(y)} \int_{y^1}^{x^1} \frac{d}{dt}\left(u(y, t,x^2 \ldots, x^{N-1}) \right) dt\right)^2 +\ldots\\&&  \ldots+N\left(  \frac{1}{f(y)}  \int_{y^{N-1}}^{x^{N-1}} \frac{d}{dt}\left(u(y, y^1,y^2 \ldots, t) \right) dt\right )^2.
\eean
Now we multiply by $ f(x)^2 w(x) f(y)^2 w(y)$ and integrate over $\Pi$ in both $x$ and $y$ variables. Taking the terms one at a time, we find for the first term
\be
\int_{\Pi\times \Pi} \frac{u(x, x^a)^2}{f(x)^2}  f(x)^2 w(x) f(y)^2 w(y) dx dx^a dy dy^a = \norm{u}{L^2(\Pi; w)}^2 \norm{f}{L^2(\Pi; w)}^2,
\ee
and the same for the second term. For the third term, we have 
\be
\int_{\Pi\times \Pi} \frac{u(x, x^a) u(y, y^a)}{f(x)f(y)}  f(x)^2 w(x) f(y)^2 w(y) dx dx^a dy dy^a = \ip{u}{f}{L^2(\Pi; w)}^2.
\ee
Now let us consider the right hand side. We first deal with the term normal to the boundary
\bean
I_0&:=& \left( \int_y^x \frac{d}{dt}\left(\frac{u(t, x^1, \ldots, x^{N-1})}{f(t)} \right) dt \right)^2\\&& \qquad\leq \abs{\int_{x_0}^{x_0+L} \frac{1}{w(t) f(t)^2} dt}  \int_{x_0}^{L+x_0}  \left[ \frac{d}{dt}\left(\frac{u(t, x^1, \ldots, x^{N-1})}{f(t)} \right)\right]^2 w(t) f(t)^2 dt .
\eean
From here, we deduce
\be
\int_{\Pi\times \Pi}I_0 f(x)^2 w(x) f(y)^2 w(y) dx dx^a dy dy^a \leq h(L, x_0) \norm{f}{L^2(\Pi; w)}^2 \norm{\df_0 u}{L^2(\Pi; w)}^2.
\ee
Similarly, we can estimate 
\be
I_1 := \left( \frac{1}{f(y)} \int_{y^1}^{x^1} \frac{d}{dt}\left(u(y, t,x^2 \ldots, x^{N-1}) \right) dt\right)^2 \leq \frac{L}{f(y)^2} \int_0^L \left[ \frac{d}{dt}\left(u(y, t,x^2 \ldots, x^{N-1}) \right)\right]^2 dt,
\ee
so that
\be
\int_{\Pi\times \Pi}I_1 f(x)^2 w(x) f(y)^2 w(y) dx dx^a dy dy^a \leq  L^2 \norm{f}{L^2(\Pi; w)}^2 \norm{\df_1 u}{L^2(\Pi; w)}^2.
\ee
Combining all of these estimates, recalling our assumption that $\norm{f}{L^2(\Pi; w)}$ is finite,  we arrive at \eq{comp1} for $u$ such that $\frac{u}{f} \in C^\infty(\overline{\Pi})$. Invoking Property A, we conclude that this in fact holds for any $u \in H^1(\Pi; w, f)$ by approximation.
\end{proof}

\begin{Lemma}
Suppose that $w, f$ have Property A on the domain $B_+^\delta$, and property $B$ with some $\epsilon>2 \delta$. Suppose $\{u^m\}$ is a bounded sequence in $H^1(B_+^\delta; w, f)$ of functions which vanish near the curved boundary of $B_+^\delta$ and which converge weakly to some $u \in H^1(B_+^\delta; w, f)$. Then $\{ u^m\}$ converges strongly in $L^2(B_+^\delta; w)$.
\end{Lemma}
\begin{proof}
We can consider $u_m$ as elements of $H^1(\Pi_0; w, f)$, where $\Pi_0 = \{(x, x^a): 0\leq x \leq 2 \epsilon / 3, \ -2 \epsilon / 3\leq x^a \leq 2 \epsilon / 3\}$. We partition $\Pi_0$ into a finite number of smaller cubes $\Pi_k$ of side length $L$ with base $x_0^k$. On each $\Pi_k$, $f, w$ have Property A so we may apply Lemma \ref{Poin} on each cube to deduce that
\be
\norm{u^m - u^n}{L^2(\Pi_0; w)}^2 \leq \sum_k \left[ \frac{\ip{u^m-u^n}{f}{L^2(\Pi_k; w)}^2}{\norm{f}{L^2(\Pi_k; w)}^2} + C(L, x^k_0) \sum_i \norm{\df_i u^m-\df_i u^n}{L^2(\Pi_k; w)}^2 \right].
\ee
Now using Property B together with the boundedness of  $\{ u^m\}$ in $H^1(\Pi_0; w, f)$, given $\tilde \epsilon>0$, we may by taking $L$ small enough assume that 
\be
\sum_k\left[ C(L, x_0^k) \sum_i \norm{\df_i u^m-\df_i u^n}{L^2(\Pi_k; w)}^2\right]<\frac{\tilde{\epsilon}}{2},
\ee
for all $m, n$. Having fixed the partition, since we know that $\{u^m\}$ converges weakly, we may by taking $m, n$ large enough make 
\be
 \sum_k \left[ \frac{\ip{u^m-u^n}{f}{L^2(\Pi_k; w)}^2}{\norm{f}{L^2(\Pi_k; w)}^2}\right]< \frac{\tilde{\epsilon}}{2}.
\ee
Thus $\{ u^m\}$ is a Cauchy sequence in $L^2(\Pi_0; w)$. By restriction, it is clearly also a Cauchy sequence in $L^2(B_+^\delta; w)$.
\end{proof}

\begin{Theorem}\label{partun}
Suppose that $\mathcal{M}$ is a manifold with boundary $\partial \mathcal{M}$ which can be covered with a finite number of coordinate charts which are either of the form
\be
\phi_I : \mathcal{U}_I \to B_+^{\delta_I},
\ee
for open coordinate patches $\mathcal{U}_I \subset \mathcal{M}$ which intersect the boundary, or else
\be
\varphi_J : \mathcal{V}_J \to  B^{\delta_J},
\ee
for open coordinate patches $\mathcal{V}_I \subset \mathcal{M}$ which do not intersect the boundary. Here $B^\delta$ is the ball of radius $\delta$ in $\mathbb{R}^N$. We assume $C^\infty$ compatibility conditions between the coordinate charts.

Suppose that $\Hh^1(\mathcal{M}), \Ll^2(\mathcal{M})$ are two Hilbert spaces of measurable functions on $\mathcal{M}$ with respective norms  $\norm{\cdot}{\Hh^1(\mathcal{M})}, \norm{\cdot}{\Ll^2(\mathcal{M})}$ such that
\begin{enumerate}[i)]
\item For each $I$, there exists a $C_I>0$ such that if $\textrm{supp } u \subset \mathcal{U}_I$ then
\be
C_I^{-1}\norm{u}{\Hh^1(\mathcal{M})} \leq \norm{u\circ \phi_I^{-1}}{H^1(B^{\delta_I}_+; w_I, f_I)} \leq C_I\norm{u}{\Hh^1(\mathcal{M})}
\ee
\be
C^{-1}_I\norm{u}{\Ll^2(\mathcal{M})} \leq \norm{u\circ \phi_I^{-1}}{L^2(B^{\delta_I}_+; w_I)} \leq C_I\norm{u}{\Ll^2(\mathcal{M})}
\ee
where $w_I, f_I$ satisfy property $A$ on $B^{\delta_I}_+$ and property $B$ with some $\epsilon_I>2 \delta_I$ for each $I$.

\item For each $J$, there exists a $C_J>0$ such that if $\textrm{supp } u \subset \mathcal{V}_J$ then
\be
C_J^{-1}\norm{u}{\Hh^1(\mathcal{M})} \leq \norm{u\circ \varphi_I^{-1}}{H^1(B^{\delta_J})} \leq C_J\norm{u}{\Hh^1(\mathcal{M})}
\ee
\be
C_J^{-1}\norm{u}{\Ll^2(\mathcal{M})} \leq \norm{u\circ \varphi_I^{-1}}{L^2(B^{\delta_J})} \leq C_J\norm{u}{\Ll^2(\mathcal{M})}
\ee
\end{enumerate}

Then $\Hh^1(\mathcal{M})$ is compactly embedded in $\Ll^2(\mathcal{M})$.

\end{Theorem}
\begin{proof}
Suppose $\{u^m\}$ is a bounded sequence in $\Hh^1(\mathcal{M})$, which we may assume without loss of generality converges weakly to $u \in \Hh^1(\mathcal{M})$. It will suffice to show that the convergence is strong in $\Ll^2(\mathcal{M})$. Let $\{\zeta_I, \zeta_J\}$ be a smooth partition of unity subordinate to $\mathcal{U}_I, \mathcal{V}_J$. It is straightforward to check\footnote{we suppress here explicit mention of the homeomorphims $\phi_I$, $\varphi_J$ for clarity} that $\{\zeta_J u^m\}$ is a bounded sequence in $H^1(B^{\delta_J})$ which converges weakly to $\{\zeta_J u\}$. By the Rellich-Kondrachov theorem, $\zeta_J u^m \to \zeta_J u$ in $L^2(B^\delta_J)$ and hence in $\Ll^2(\mathcal{M})$. Similarly, $\{\zeta_I u^m\}$ is a bounded sequence in $H^1(B^{\delta_I}_+; w_I, f_I)$ which converges weakly to $\{\zeta_I u\}$. By Theorem \ref{partun}, $\zeta_I u^m \to \zeta_I u$ in $L^2(B^{\delta_I}_+; w_I)$ and hence in $\Ll^2(\mathcal{M})$. Now, taking the (finite) sum over the partition of unity, we conclude that $u^m \to u$ in $\Ll^2(\mathcal{M})$ and we are done.
\end{proof}

We have thus seen that Properties A and B, applied locally at the boundary together with some form of compactness, are sufficient to imply the compact embedding which we require. We note that we have not shown that these properties are necessary, but an examination of our proof suggests that we cannot easily weaken them and retain the same method of proof. We now wish to give some conditions under which Properties A and B will hold. Let us denote 
\be
s(t) = f(t)^2 w(t),
\ee
and we can assume that $s(t)$ is defined on some interval $[0, T)$. By our previous assumptions, $s(t)$ is smooth and positive in $(0, T)$.

\begin{Lemma}
Suppose
\be
\lim_{t \to 0} s(t) = s_0, \qquad s_0 >0,
\ee
then for $\delta$ sufficiently small, Property A holds on $B^\delta_+$ and B holds for some $\epsilon>2\delta$.
\end{Lemma}
\begin{proof}
We first note that if $s(t)$ tends to a finite, non-zero, limit as $t \to 0$, then we have the equivalence of the norms:
\be
\norm{u}{H^1(B^\delta_+; w, f)} \sim \norm{\frac{u}{f}}{H^1(B^\delta_+)}.
\ee
Property A then follows immediately from the density of $C^\infty(B^\delta_+)$ in $H^1(B^\delta_+)$. Furthermore, Property B follows from the fact that both integrands in the definition of $h(L, x_0)$ belong to $C^0[0, T)$, so we may take $\epsilon=T/3$ and $\delta<T/6$.
\end{proof}

\begin{Lemma}
Suppose $s(t)$ is non-decreasing on $(0, T)$, and suppose
\be
\lim_{t \to 0} s(t) = 0,
\ee
and
\be
\int_0^T \frac{1}{s(t)} dt < \infty,
\ee
then for $\delta$ sufficiently small, Property A holds on $B^\delta_+$ and B holds for some $\epsilon>2\delta$.
\end{Lemma}
\begin{proof}
Property A follows from \cite[Theorem 11.2]{Kufner}. Property B follows since we may estimate for $0\leq x_0 < T/3$, $0<L<T/3$
\be
\abs{\left(\int_{x_0}^{L+x_0} \frac{1}{w(t) f(t)^2} dt\right)\left(\int_{x_0}^{L+x_0}{w(t) f(t)^2} dt\right)} \leq L \norm{s}{L^\infty[0, 2T/3]}\int_0^{2 T/3} \frac{1}{s(t)} dt .
\ee
Again, we may take $\epsilon=T/3$ and $\delta<T/6$.
\end{proof}

\begin{Lemma}
Suppose $s(t)$ is non-increasing in the interval $(0, T)$, and suppose
\be
\lim_{t \to 0} s(t) = \infty,
\ee
and
\be
\int_0^T {s(t)} dt < \infty,
\ee
then for $\delta$ sufficiently small, Property A holds on $B^\delta_+$ and B holds for some $\epsilon>2\delta$.
\end{Lemma}
\begin{proof}
Property A follows from \cite[Lemma 11.8]{Kufner}. Property B follows since we may estimate for $0\leq x_0 < T/3$, $0<L<T/3$
\be
\abs{\left(\int_{x_0}^{L+x_0} \frac{1}{w(t) f(t)^2} dt\right)\left(\int_{x_0}^{L+x_0}{w(t) f(t)^2} dt\right)}  \leq L \norm{s^{-1}}{L^\infty[0, 2T/3]}\int_0^{2 T/3} {s(t)} dt .
\ee
Again, we may take $\epsilon=T/3$ and $\delta<T/6$.
\end{proof}

We state some explicit results:
\begin{Theorem}\label{cmpct}
\begin{enumerate}[i)]
\item Let
\be
w(t) = t^k, \qquad f_m(t) = t^m.
\ee
Then if $-1<k+2 m<1$, for any $\delta$, Property A holds on $B^\delta_+$ and B holds for some $\epsilon>2\delta$.
\item Let
\be
w(t) = t, \qquad f(t) = \abs{\log(t)}.
\ee
Then for $\delta<1/12$, Property A holds on $B^\delta_+$ and B holds for $\epsilon=1/6$.
\end{enumerate}
\end{Theorem}
\begin{proof}
We simply apply the various results of this section to the stated weight and twisting functions.
\end{proof}

These proofs are sufficient, when applied to $\Hh^1(\Sigma, \kappa)$ to deal with the interval $\frac{5}{4}<\alpha<\frac{9}{4}$. For the case $\alpha\leq \frac{5}{4}$, where we necessarily take Dirichlet boundary conditions at infinity, we require the following Lemma. 
\begin{Lemma}\label{propAB}
Let
\be
w(t) = t^k, \qquad f_m(t) = t^m.
\ee
Then the spaces $H^1_0(B^\delta_+; w, f_m)$ are equivalent for any $m$ with $k+2m \neq 1$.
\end{Lemma}
\begin{proof}
We first suppose that $u \in C_c^\infty(B^\delta_+)$. Consider
\bean
\norm{\frac{u}{x}}{L^2(B^\delta_+; w)} &=& \int_{\mathbb{R}^{N-1}}dx^a \int_0^1 dx u^2 x^{k-2} = \int_{\mathbb{R}^{N-1}}dx^a \int_0^1 dx (x^{-m} u)^2 \cdot x^{k-2+2m} \\ \nonumber &=& \frac{2}{k-1+2m} \int_{\mathbb{R}^{N-1}}dx^a \int_0^1 dx\  x^{k-1}  u \left(x^{m}\frac{\partial}{\partial x} x^{-m} u \right) \nonumber \\ &\leq &C_\delta  \int_{\mathbb{R}^{N-1}}dx^a \int_0^1 dx\ x^k \left(x^{m}\frac{\partial}{\partial x} x^{-m} u \right)^2 +\delta \int_{\mathbb{R}^{N-1}}dx^a \int_0^1 dx  u^2 x^{k-2}.
\eean
Choosing $\delta$ appropriately, we conclude that $\norm{u/x}{L^2(B^\delta_+; w)} \leq C \norm{u}{H^1(B^\delta_+; w, f_m)}$, where the constant depends on $U$, $m$ and $k$. Noting now that
\be
 \left(x^{m'}\frac{\partial}{\partial x} x^{-m'} u \right) =  \left(x^{m}\frac{\partial}{\partial x} x^{-m} u \right)+ (m-m') \frac{u}{x}
\ee
we immediately conclude that there exist $c, C$ depending on $U$, $k$, $m$ such that for any $u \in C_c^\infty(B^\delta_+)$ we have:
\be
c \norm{u}{H^1(B^\delta_+; w, f_m)}\leq \norm{u}{H^1(B^\delta_+; w, f_{m'})} \leq C \norm{u}{H^1(B^\delta_+; w, f_m)}.
\ee
By approximation, the result follows.
\end{proof}

We are now ready to prove the main result. 
\begin{proof}[Proof of Theorem  \ref{RKthm}]
First, consider a neighbourhood of conformal infinity. We recall that this neighbourhood may be covered with a finite number of coordinate patches $\tilde{\mathcal{U}}_{\tilde{I}}$ with coordinates $(s, x^A)$, where $s=0$ is the conformal boundary. We may assume without loss of generality that $(s, x^A)\in B^{\delta_{\tilde{I}}}_+$ for some $\delta_{\tilde{I}}$. We have from Section \ref{waveop}
\be
\begin{array}{rclcrcl}
A &=& s^{-2} \tilde A_0 + \O{1}, & \qquad & h &=& s^{-6} \tilde h_0 + \O{s^{-4}},\\
 a^{ss} &=& \tilde A_0 +\O{s^2}, & \qquad&  a^{AB} &=& \tilde A_0 \tilde\sigma_0^{AB}+ \O{s^2}, \\
&&& a^{s A} = \O{s^2},&&&
\end{array}
\ee
where $\tilde{A}_0>0$ and $\tilde{\sigma}_0^{AB}$ are positive definite. We also have $f =f_\kappa:= s^{\frac{3}{2}-\kappa}$ near $s=0$. From here, we deduce  that there exist constants $C_{\tilde{I}}>0$ such that for any $u$ supported in $\tilde{\mathcal{U}}_{\tilde{I}}$ we have
\be
C_{\tilde{I}}^{-1} \norm{u}{\Hh^1(\Sigma, \kappa)}^2 \leq  \norm{u}{H^1(B^{\delta_{\tilde{I}}}_+; w, f_{\kappa})}^2 \leq C_{\tilde{I}} \norm{u}{\Hh^1(\Sigma, \kappa)}^2,
\ee
\be
C_{\tilde{I}}^{-1} \norm{u}{\Ll^2(\Sigma)}^2 \leq  \norm{u}{L^2(B^{\delta_{\tilde{I}}}_+; w)}^2 \leq C_{\tilde{I}} \norm{u}{\Ll^1(\Sigma, \kappa)}^2,
\ee
where $w = s^{-2}, f_{\kappa} = s^{\frac{3}{2}-\kappa}$.

Now recall that a neighbourhood of the horizon may be covered by a finite number of coordinate patches $\mathcal{U}_I$ with coordinates $(\rho, x^A)$, where $\rho=0$ is the horizon. We have
\be
\begin{array}{rclcrcl}
A &=& A_0 + \O{\rho}, & \qquad & h &=& h_0 + \O{\rho},\\
a^{\rho \rho} &=& C \rho + \O{\rho^2},& \qquad&  a^{A B} &=& A_0 \sigma_0^{A B}+ \O{\rho}, \\
&&& a^{\rho A} = \O{\rho},&&&
\end{array}
\ee
where $C, {A}_0>0$ and ${\sigma}_0^{AB}$ are positive definite. We also have $f = \abs{\log \rho}$ near $\rho=0$. We make the change of variables $\rho =  t^2$, and find
\be
\begin{array}{rclcrcl}
A &=& A_0 + \O{t^2}, & \qquad & h &=& t h'_0 + \O{t^3},\\
a^{tt} &=& C'  + \O{t^2},& \qquad&  a^{A B} &=& A_0 \sigma_0^{A B}+ \O{t^2}, \\
&&& a^{t A} = \O{t},&&&
\end{array}
\ee
and $f = 2\abs{\log t}$ near $t=0$. By refining our cover if necessary, we may assume that the image of $\mathcal{U}_I$ in these coordinates if $B^{\delta_I}_+$ for some $\delta_I<1/12$. From here, we deduce that there exist constants $C_I>0$ such that for any $u$ supported in $\mathcal{U}_I$ we have
\be
C_I^{-1} \norm{u}{\Hh^1(\Sigma, \kappa)}^2 \leq  \norm{u}{H^1(B_+^{\delta_I}; w, f)}^2 \leq C_I \norm{u}{\Hh^1(\Sigma, \kappa)}^2,
\ee
\be
C_I^{-1} \norm{u}{\Ll^2(\Sigma)}^2 \leq  \norm{u}{L^2(B^{\delta_I}_+; w)}^2 \leq C_I \norm{u}{\Ll^1(\Sigma, \kappa)}^2,
\ee
where $w(t) = t, f(t) = \abs{\log t}$.

Now, since $\Sigma$ with the $\mathcal{U}_I$'s and $\tilde{\mathcal{U}}_{\tilde{I}}$'s removed is compact, we can the remainder of $\Sigma$ with a finite number of coordinate patches $\mathcal{V}_J$ whose image under the coordinate map is $B^{\delta_J}$ and such that $\Hh^1(\Sigma, \kappa)$ is equivalent to $H^1(B^{\delta_J})$ for functions supported in $\mathcal{V}_J$. Thus, taking into account Theorem \ref{cmpct}, we have verified that the conditions of Theorem \ref{partun} hold for $\Hh^1(\Sigma, \kappa)$ and $\Ll^2(\Sigma)$ provided $0<\kappa<1$. For $\kappa> 0$, making use of Lemma \ref{propAB} we can verify that the conditions of Theorem \ref{partun} hold for $\Hh^1_0(\Sigma, \kappa)$ and $\Ll^2(\Sigma)$.
\end{proof}

\appendix

\section{The method of counter terms \label{counter}}

We discuss in this appendix the counter-term renormalization for the energy-momentum tensor introduced by Breitenlohner and Freedman \cite{BF}. They added the following counter term to the energy-momentum tensor
\ben{ctm}
\hat T_{\mu \nu} = g_{\mu\nu} \Box_g \psi^2 - \nabla_{\mu}\nabla_{ \nu} \psi^2 + R_{\mu \nu} \psi^2 \, .
\een
If $K$ is a Killing field, then we claim that the flux of $\hat J^K_\mu = \hat T_{\mu \nu} K^\nu$ through any surface $S$ can be expressed as an integral over $\partial S$. In order to see this, we make use of the following fact about integration over differential forms of degree one:
\ben{inteq}
\int_S \beta_\mu dS^\mu = \int_S \iota^*(\star \beta) ,
\een
here $\beta = \beta_\mu dx^\mu$ is a differential form and $dS^\mu$ is the normal volume element induced on a co-dimension one surface $S$ by the metric $g$, whose Hodge star is $\star$. Now, we re-write $\hat J_\nu^K $ as follows:
\bea
\hat J_\mu^K &=& \nonumber K_\mu \nabla^\nu \nabla_\nu \psi^2 - K_\nu \nabla^\nu \nabla_\mu \psi^2 + R_{\mu}{}^\nu K_\nu\psi^2 \\ &=& \nabla^\nu(K_\mu \nabla_\nu \psi^2 - K_{\nu}\nabla_\mu \psi^2) - \nabla^\nu \psi^2 \nabla_\nu K_\mu + \psi^2 \nabla^\nu \nabla_\nu K_\mu
 \\ &=& \nabla^\nu(K_\mu \nabla_\nu \psi^2 - K_{\nu}\nabla_\mu \psi^2) + \frac{1}{2} \nabla^\nu( \psi^2 \nabla_\nu K_\mu -\psi^2 \nabla_\mu K_\nu ). \nonumber
\eea
Here we have used two properties of Killing vectors. Firstly, we use Killing's equation:
\be
\nabla_\mu K_\nu = \frac{1}{2}(\nabla_\mu K_\nu -\nabla_\nu K_\mu ),
\ee
and we also require the following consequence which comes from differentiating Killing's equation and doing some index shuffling:
\be
 \nabla^\nu \nabla_\nu K_\mu = R_{\mu}{}^\nu K_\nu.
\ee
Thus, up to a constant multiple we have
\be
\hat J_\mu^K dx^\mu = \delta(K^\flat \wedge d(\psi^2) +  \psi^2 dK^\flat) = \delta d (\psi^2 K^\flat).
\ee
Inserting this into \eq{inteq}, we have (up to factors)
\be
\int_S \hat J_\mu^K dS^\mu = \int_S \iota^\star [d \star d (\psi^2 K^\flat)] = \int_S d \eta,
\ee
where
\be
\eta = \iota^\star[\star d (\psi^2 K^\flat)] 
\ee
From here we conclude that the integral of $\hat J_\mu^K$ over any closed surface must vanish.  
Furthermore we see that
\be
\int_S \hat J_\mu^K dS^\mu = \pm \int_{\partial S} \partial_{\mu} (\psi^2 K_{\nu}) dS^{[\mu \nu]}.
\ee 
Here $dS^{[\mu \nu]} = n_1^{[\mu} n_2^{\nu]} dS_{\partial S}$ where $n_1^\mu$ is the unit normal to $S$ and $n_2^\mu$ is the unit normal of $\partial S$ considered as a submanifold of $S$. The undetermined sign can be fixed by making a choice of orientations. Note that this calculation made no assumptions on the metric other than that it admits a Killing field. This result generalises equation (6.13) of Breitenlohner and Freedman \cite{BF}.

\subsection{The modified fluxes for AdS-Schwarzschild}

After performing the integrations by parts, we find the following expressions for the energy fluxes in an AdS-Schwarzschild background due to the Breitenlohner-Freedman modification of the energy-momentum tensor:
\bea
\nonumber \int_{\Sigma_t^{[R_1, R_2]}} \hat J^T_\mu n^\mu dS_{\Sigma_t} &=& \int_{S^2_{t, R_2}} \left[ -2g^{rr} \psi \left( \tn_r\psi \right)-2g^{rt} \psi \nabla_t \psi + \overline{S}(r) \psi^2 \right] r^2 d\omega \\ && \quad - \int_{S^2_{t, R_1}} \left[ -2g^{rr} \psi \left( \tn_r\psi \right)-2g^{rt} \phi \nabla_t \psi + \overline{S}(r) \psi^2 \right] r^2 d\omega,
\eea
and
\bea
\nonumber \int_{\tilde{\Sigma}_r^{[T_1, T_2]}} \hat J^T_\mu m^\mu dS_{\tilde{\Sigma}_r} &=& \int_{S^2_{T_2, r}} \left[ -2g^{rr} \psi \left( \tn_r \psi \right)-2g^{rt} \psi \nabla_t \psi + \overline{S}(r) \psi^2 \right] r^2 d\omega \\ && \quad - \int_{S^2_{T_1, r}} \left[ -2g^{rr} \psi \left( \tn_r \psi \right)-2g^{rt} \psi \nabla_t \psi +\overline{S}(r) \psi^2 \right] r^2 d\omega,
\eea
where $\overline{S}(r)$ is given by
\be
\overline{S}(r) = \frac{1}{2}\partial_r(g^{rr})- 2 g^{rr} \frac{f'(r)}{f(r)}.
\ee
Now for large $r$, assuming we take $f(r)\sim r^{-\frac{3}{2} + \kappa}(1+ \O{r^{-2}})$ we have
\be
\overline{S}(r)  = \frac{2 (2-\kappa)}{l^2} r + \O{\frac{1}{r}},
\ee
whereas recalling the surface term for the unmodified energy fluxes \eq{Sdef} we have
\be
S(r)  = -\frac{ (3-2\kappa)}{4 l^2} r + \O{\frac{1}{r}}.
\ee
Thus if we define
\be
\overline{T}_{\mu \nu} = T_{\mu \nu} + \frac{8(2-\kappa)}{3-2\kappa} \hat T_{\mu \nu},
\ee
we render the fluxes of currents $\overline{J}^T_\mu = \overline{T}_{\mu \nu} T^\nu$ through the surfaces $\Sigma_t^{[R_1, R_2]}$ and $\Sigma_{R_2}^{[T_1, T_2]}$ finite as $R_2 \to \infty$. This would seem to be good news, however it comes at the price of introducing surface terms on the horizon. We in fact have
\bea
\nonumber \int_{\Sigma_t^{[r_{+}, \infty)}} \overline{J}^T_\mu n^\mu dS_{\Sigma_t} &=& E(t) -\int_{S^2_{t, r_{+}}} \left[ -2g^{rt} \psi \nabla_t \psi + \overline{S}(r) \psi^2 \right] r^2 d\omega, \\
 \int_{\tilde{\Sigma}_{r_{+}^{[T_1, T_2]}}} \overline{J}^T_\mu m^\mu dS_{\tilde{\Sigma}_r} &=& F[T_1,T_2] +\int_{S^2_{T_2, r_{+}}} \left[ -2g^{rt} \psi \nabla_t \psi + \overline{S}(r) \psi^2 \right] r^2 d\omega \\ && \qquad - \int_{S^2_{T_1, r_{+}}} \left[ -2g^{rt} \psi \nabla_t \psi + \overline{S}(r) \psi^2 \right] r^2 d\omega, \nonumber \\
\lim_{R_2 \to \infty} \int_{\tilde{\Sigma}_{R_2}^{[T_1, T_2]}}\overline{J}^T_\mu m^\mu dS_{\tilde{\Sigma}_r}&=&0, \nonumber
\eea
assuming either Neumann or Dirichlet conditions on $\scri$. The advantage of the counter term method is that we may now directly apply the divergence theorem to the whole infinite slab we are interested in, since now the fluxes are all finite. The disadvantage is that we pick up the undesirable terms on the horizon. These cancel in the energy identity (of course), but are rather inelegant.

\end{document}